%% file: main.tex
\documentclass[lettersize,journal]{IEEEtran}

\input{preamble}

\begin{document}
\bstctlcite{IEEEexample:BSTcontrol} 

\title{Fast Tri-Hybrid Beamforming via Deep Unfolding
}

\author{Pinjun~Zheng, Md.~Jahangir Hossain, and Anas~Chaaban \vspace{-1em}

\thanks{The authors are with the School of Engineering, The University of British Columbia, Kelowna, BC V1V 1V7, Canada (e-mail: pinjun.zheng@ubc.ca; jahangir.hossain@ubc.ca; anas.chaaban@ubc.ca).

} 
}

\maketitle

\begin{abstract}
Tri-hybrid multiple-input multiple-output architectures have recently emerged as a promising enabler for next-generation wireless systems, as they potentially provide enhanced design flexibility without a proportional increase in hardware cost or power consumption. However, the resulting triple-domain coupling renders beamforming optimization challenging and computationally demanding. This paper develops a fast tri-hybrid beamforming framework for multiuser downlink systems employing dynamic metasurface antennas (DMAs). Based on the equivalence between weighted sum-rate maximization and weighted sum-minimum mean square error minimization, an iterative algorithm is first derived under per-DMA input power constraints, with all update equations available in closed form, but convergence inherently remains slow. To enable real-time operation, the algorithm is further unfolded into a trainable finite-iteration architecture using graph neural networks that ensure permutation equivariance and support varying numbers of users. Trained on ray-tracing channel data, the unfolded method achieves comparable or higher system sum-rates while reducing runtime by more than an order of magnitude. The method also demonstrates strong scalability, robustness, and generalization across various environments.
\end{abstract}

\section{Introduction}

Tri-hybrid \ac{MIMO} has recently emerged as a novel transceiver architecture with the potential to serve as a cornerstone of next-generation \ac{MIMO} communication systems~\cite{Heath2025Tri,Castellanos2025Embracing}. It extends the well-known hybrid digital-analog \ac{MIMO} architecture~\cite{Alkhateeb2014MIMO,Sohrabi2016Hybrid,Ioushua2019Family} by incorporating a third layer consisting of a reconfigurable antenna array or tunable \ac{RF} front end, commonly referred to as the \emph{antenna domain} or \emph{\ac{EM} domain}~\cite{Ying2024Reconfigurable,Liu2025Reconfigurable,Zheng2025Reconfigurable}. The main advantage of such a tri-hybrid architecture lies in providing additional antenna-domain control without proportionally increasing \ac{RF} complexity or power consumption~\cite{Heath2025Tri,Castellanos2023Energy,Li2026Tri}. This offers enhanced design flexibility for meeting the stringent performance requirements of next-generation wireless networks.

The antenna domain reconfigurability can be implemented in various ways, depending on the employed antenna technologies. Representative examples include \acp{DMA}\cite{PhysRevApplied}, pixel antennas\cite{Shen2017Successive}, lens antennas~\cite{Gao2019Wideband}, and fluid antennas~\cite{Wang2026Electromagnetically}.  These technologies enable dynamic adaptation of antenna characteristics, such as radiation pattern, polarization, operating frequency, and other \ac{EM} properties~\cite{Costantine2015Reconfigurable}. Although these technologies have been extensively studied in the antenna community, their integration into tri-hybrid \ac{MIMO} remains largely underexplored. {Notably, recent studies have shown that tri-hybrid MIMO architectures employing \acp{DMA} can achieve a more favorable trade-off between spectral and energy efficiency by jointly leveraging the distinct hardware characteristics of the digital, analog, and antenna domains~\cite{Castellanos2023Energy, Castellanos2025Embracing}.} An early study also demonstrates enhanced hardware efficiency when pattern-reconfigurable antennas are incorporated~\cite{Zheng2025Trihybrid}. Moreover, the use of multifunctional compound reconfigurable antennas in tri-hybrid systems is under active investigation~\cite{Liu2025Secure}. {Indeed, additional reconfigurability had already been explored in conventional hybrid \ac{MIMO} architectures before the advent of tri-hybrid \ac{MIMO}, through mechanisms such as dynamic \ac{RF}-chain selection and dynamic subarray connections designed to improve hardware efficiency and beamforming performance~\cite{Kaushik2019Dynamic,Park2017Dynamic,Li2020Dynamic}. However, these mechanisms primarily reconfigure the analog-domain feed network, whereas tri-hybrid \ac{MIMO} introduces an additional and distinct antenna-domain processing layer that can be configured independently while operating complementarily with the digital and analog domains.}

Despite these advantages, tri-hybrid architecture also complicates the beamforming design, as it involves the joint optimization of digital, analog, and \ac{EM} precoders under practical hardware constraints, for which no simple solution exists. Several studies have sought to address the tri-hybrid beamforming problem, the vast majority of which adopt alternating optimization frameworks that iteratively update the precoders across domains by combining various optimization techniques. For example, the authors of~\cite{Li2026Tri} proposed an alternating optimization method that integrates penalty dual decomposition and coordinate descent methods for tri-hybrid beamforming based on radiation-center reconfigurable antennas. In~\cite{Chen2026Tri}, a solution for pinching-antenna-assisted multiuser downlink was developed by combining \ac{ZF} beamforming with probabilistic learning. The work in~\cite{Zhang2026Trihybrid} addressed tri-hybrid beamforming for integrated sensing and communication using an alternating scheme based on sequential quadratic programming and adaptive gradient projection. While the technical details of these approaches vary widely due to differing optimization objectives and system models, they all rely on the coordinated use of multiple optimization techniques. This complexity is inherent to tri-hybrid beamforming, arising from the coupled design variables and constraints introduced by the multi-domain architecture.

In theoretical studies of beamforming algorithms for such complex systems, computational efficiency is easily overlooked when primary attention is placed on asymptotic convergence performance. In practice, however, beamforming in \ac{MIMO} communications is inherently latency-sensitive, as the channel coherence time is typically limited and real-time reconfiguration is required~\cite{Wiesmayr2025Design}. This challenge is particularly pronounced in tri-hybrid beamforming, where most algorithms rely on numerous iterations or alternating updates across the three precoding domains, which can significantly slow down convergence. In fact, a few existing works in tri-hybrid beamforming have already emphasized computational efficiency. For instance, the algorithms in~\cite{Fang2026Trihybrid} and~\cite{Chen2025Integrated} reduce per-iteration complexity by designing closed-form updates. Nonetheless, low-complexity updates alone do not fully resolve the computational burden, as the overall cost also depends on convergence speed. Some low-complexity methods require hundreds of iterations to converge, as reported in~\cite{Chen2025Integrated, Chen2026Tri}, which can still lead to substantial runtime in practice.

To fill this gap, \emph{deep unfolding} (or \emph{unrolling}) provides a principled approach to accelerating iterative algorithms by embedding trainable components into the iterative procedure~\cite{Monga2021Algorithm}. While it does not simplify individual updates, the added trainable components can enhance the update rules and achieve improved performance within fewer iterations after proper training. Deep unfolding has been successfully applied to a range of related problems, such as interference network power allocation~\cite{Chowdhury2021Unfolding}, joint radar-communication beamforming~\cite{Zhang2025Low}, and more~\cite{Khobahi2021LoRD, Jagannath2021Redefining}. Motivated by these advances, this paper advocates deep unfolding as an effective means of accelerating tri-hybrid beamforming without compromising performance. We focus on a tri-hybrid beamforming problem based on \acp{DMA} in a multiuser downlink scenario. Targeting real-time communications, our \emph{goal} is to develop an algorithm that provides a high-quality tri-hybrid precoding design rapidly. 

The proposed solution consists of two main steps: we first \emph{(i)} develop a low-complexity (but potentially slow-converging) iterative algorithm, and then \emph{(ii)} design a deep unfolding network to accelerate its convergence. Although this paper focuses on a specific antenna type, the presented paradigm can be readily adapted to tri-hybrid beamforming problems under different system configurations. Specifically, our contributions are summarized as follows:  
\begin{itemize}
	\item We develop an iterative algorithm for the downlink sum-rate maximization problem in a multiuser tri-hybrid \ac{MIMO} system employing \acp{DMA}. The algorithm is derived based on the \ac{WMMSE} equivalence principle combined with a few alternating optimization steps. Each update equation is provided in closed-form to ensure low per-iteration complexity, although the overall runtime remains considerable due to slow convergence caused by the extensive alternating updates.
	
	\item We propose a data-driven method based on deep unfolding to accelerate convergence. Specifically, we introduce a set of trainable components into the iterative procedure, with parameters learned from user \ac{CSI} features via a deep neural network. In particular, an architecture based on the \ac{GraphSAGE} algorithm is developed to inductively learn the acceleration parameters from user \ac{CSI}, maintaining permutation equivariance and naturally accommodating a varying number of users.
	
	\item We train the proposed deep learning model using ray-tracing channel data generated for multiple cities and evaluate it on data collected from a different city. To ensure effective training, we design problem-specific data preprocessing and parameter initialization strategies. The proposed method is comprehensively evaluated in terms of acceleration performance, scalability, robustness to \ac{CSI} errors, and generalization capability. The training code and our trained network weights are available in~\cite{ZPinjun_FastTHBF_2026}.
\end{itemize}

The rest of the paper is organized as follows. Section~II introduces the system model and formulates the tri-hybrid beamforming problem. Section~III develops an iterative algorithm to solve the formulated problem, which, while effective, inevitably suffers from slow convergence. To address this limitation, Section~IV proposes a deep unfolding approach to accelerate the algorithm. Section~V presents numerical results. Finally, Section~VI discusses the limitations of the proposed method, and Section~VII concludes the paper.

We use the following notation throughout the paper. Non-bold lowercase and uppercase letters (e.g., $a, A$) denote scalars, bold lowercase letters (e.g., $\av$) denote vectors, and bold uppercase letters (e.g., $\Am$) denote matrices. We use $[\av]_{i}$ to denote the~$i^\text{th}$ entry of vector $\av$, and use $[\Am]_{i,j}$ to denote the entry at the~$i^\text{th}$ row and~$j^\text{th}$ column of the matrix~$\Am$. The superscripts ${(\cdot)}^\TT$, ${(\cdot)}^\HH$, and ${(\cdot)}^{-1}$ represent the transpose, Hermitian (conjugate transpose), and inverse operators, respectively. In addition, $\mathsf{j}$ denotes the imaginary unit with $\mathsf{j}^2 = -1$. The operators $\mathrm{Re}(\cdot)$ and $\mathrm{Im}(\cdot)$ extract the real and imaginary parts of a complex quantity, respectively, and $\mathrm{vec}(\cdot)$ denotes the vectorization operator that stacks the columns of a matrix into a single column vector. The operator $\mathrm{diag}(\cdot)$ constructs a diagonal matrix from its argument, while $\mathrm{blkdiag}(\cdot)$ constructs a block-diagonal matrix from its arguments.

\section{System Model}\label{sec:SACM}

\begin{figure*}[t]
  \centering
  \includegraphics[width=0.9\linewidth]{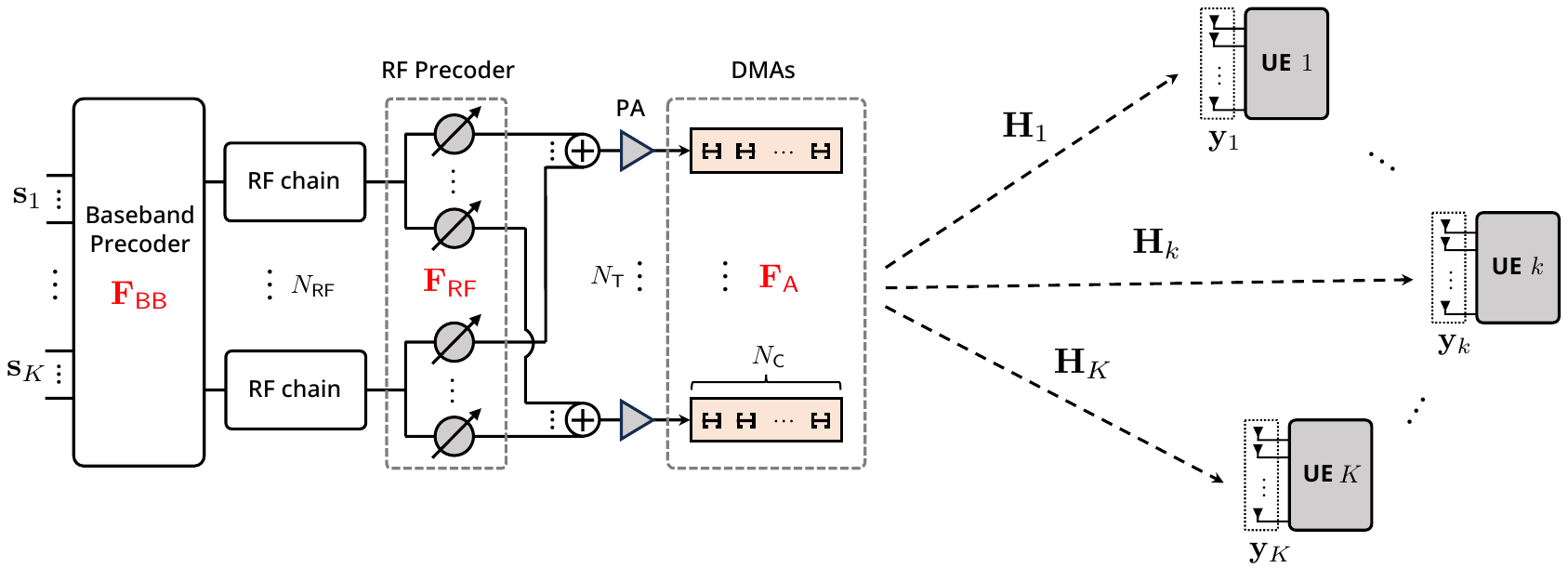}
  \vspace{-0.5em}
\caption{Schematic of a tri-hybrid MIMO downlink system with a BS equipped with $N_\mathsf{RF}$ RF chains and $N_\mathsf{T}$ DMAs serving $K$ users.}
  \label{fig_system}
\end{figure*}

We consider a multiuser tri-hybrid \ac{MIMO} downlink system, whose architecture is illustrated in Fig.~\ref{fig_system}. The transmitter is a \ac{BS} equipped with $N_\mathsf{RF}$ \ac{RF} chains and $N_\Tt$ \acp{DMA}. Each \ac{DMA} is a tunable leaky-wave antenna composed of $N_\Ct$ {quasi-passive} radiating elements, whose resonant responses can be adjusted through semiconductor tuning components, such as varactors, thereby changing the amplitude and phase of the radiated signals~\cite{PhysRevApplied}. As a result, the \ac{BS} comprises a total of $N_\Ct N_\Tt$ effective radiating elements. The \ac{BS} serves $K$ users simultaneously. The $k^\text{th}$ user is equipped with $M_k$ antennas and is allocated $N_{\St,k}$ data streams. Accordingly, the total number of transmitted downlink data streams is $N_\St = \sum_{k=1}^K N_{\St,k}$. Typically, the system dimensions satisfy $N_\St \leq N_\mathsf{RF} \leq N_\Tt$. 

\subsection{Tri-Hybrid Precoding Architecture}

Let $\sv_k\in\mathbb{C}^{N_{\St,k}}$ denote the data symbol vector intended for the $k^\text{th}$ user, and define the aggregated symbol vector as $\sv=[\sv_1^\TT,\dots,\sv_K^\TT]^\TT\in\mathbb{C}^{N_\St}$. We assume $\mathbb{E}[\sv\sv^\HH]=\mathbf{I}$. The \ac{BS} employs a three-stage precoding architecture. Specifically, the data symbols are first processed by a baseband digital precoder $\Fm_\mathsf{BB}\in\mathbb{C}^{N_\mathsf{RF}\times N_\St}$, followed by an \ac{RF}-domain analog precoder $\Fm_\mathsf{RF}\in\mathbb{C}^{N_\Tt\times N_\mathsf{RF}}$, and finally by an antenna-domain precoder $\Fm_\mathsf{A}\in\mathbb{C}^{N_\Ct N_\Tt\times N_\Tt}$ implemented by the \acp{DMA}. The resulting transmitted signal is given by $\Fm_\mathsf{A}\Fm_\mathsf{RF}\Fm_\mathsf{BB}\sv$. Assuming a \emph{narrowband} communication model, the characteristics of these three precoders are detailed below.

\subsubsection{Fully Connected RF Precoding}
We consider a fully connected RF precoding architecture, where each RF chain is connected to all \acp{DMA}. Consequently, $\Fm_\mathsf{RF}$ is a full matrix. Since the RF precoder is implemented using phase shifters, its entries are constrained to constant modulus. Without loss of generality, we adopt the normalization
$|[\Fm_\mathsf{RF}]_{i,j}|^2 = {1}/{N_\Tt},\ \forall i,j$, while the digital precoder $\Fm_\mathsf{BB}$ is unconstrained.

\subsubsection{Per-DMA Input Power Constraints}
{Since the \acp{DMA} are assumed to be quasi-passive~\cite{PhysRevApplied}, they provide no additional power gain. Accordingly, the power constraint is imposed only at their inputs.}
As illustrated in Fig.~\ref{fig_system}, the signals processed by the \ac{RF} precoder are amplified by a dedicated \ac{PA} before being fed into each \ac{DMA}. Since each \ac{PA} operates independently and has a limited linear operating range and maximum output capability, the input power to each \ac{DMA} must be individually constrained. Accordingly, we impose a \emph{per-DMA input power constraint} as
\begin{equation}
	\big[\Fm_\mathsf{RF} \Fm_{\mathsf{BB}} \Fm_{\mathsf{BB}}^\HH \Fm_\mathsf{RF}^\HH \big]_{n,n}
	\leq P_n,\ n = 1,2,\dots,N_\Tt.	
\end{equation}

\subsubsection{DMA Precoding Model} 
{Each \ac{DMA} is a leaky-wave antenna that successively radiates the guided \ac{EM} wave through its $N_\Ct$ tunable elements.  Across the \ac{DMA} array, however, the $N_\Tt$ \acp{DMA} are excited concurrently, as the digital- and analog-domain signal paths operate in parallel.} Mathematically, the antenna-domain precoding performed by the $n^\text{th}$ \ac{DMA} can be represented by applying a weight vector $\wv_n\in\mathbb{C}^{N_\Ct}$ to its input signal. As a result, the antenna-domain precoder admits the block-diagonal structure $\Fm_\mathsf{A} = \mathrm{blkdiag}(\wv_1,\wv_2,\dots,\wv_{N_\Tt})$. The weight vector $\wv_n$ is determined by the physical characteristics of the DMA. Specifically, it can be expressed as
\begin{equation}
	\wv_n = [\eta_{n,1}\alpha_{n,1},\dots,\eta_{n,N_\Ct}\alpha_{n,N_\Ct}]^\TT,
\end{equation}
where $\eta_{n,m}$ denotes the complex waveguide response from the feed point to the $m^\text{th}$ element, and $\alpha_{n,m}$ represents the tunable response of that element. Generally, $\eta_{n,m}$ can be obtained via full-wave simulations for a given hardware implementation~\cite{Carlson2024Hierarchical}, and several analytical models have also been developed to characterize it~\cite{Zhang2022Beam,Williams2022Electromagnetic,Kimaryo2023Downlink,Bayraktar2024near}. The tunable element response $\alpha_{n,m}$ can be modeled using the Lorentzian resonator model~\cite{PhysRevApplied,Carlson2024Hierarchical}, under which $\alpha_{n,m} = \frac{\mathsf{j}+e^{\mathsf{j}\phi_{n,m}}}{2}$, where $\phi_{n,m} \in [0, 2\pi)$ denotes the controllable phase shift.

Defining $\etav_n \triangleq [\eta_{n,1}, \eta_{n,2}, \dots, \eta_{n,N_\Ct}]^\TT$, $\Gm_n\triangleq\mathrm{diag}(\etav_n)$, and $\uv_n \triangleq [e^{\mathsf{j}\phi_{n,1}}, e^{\mathsf{j}\phi_{n,2}}, \dots, e^{\mathsf{j}\phi_{n,N_\Ct}}]^\TT$, the \ac{DMA} weight vector $\wv_n$ can be written in a compact form as~\cite{Kimaryo2023Downlink,Bayraktar2024near}
\begin{equation}\label{wvn}
	\wv_n = \frac{1}{2}\Gm_n\uv_n + \frac{\mathsf{j}}{2}\etav_n,\ n=1,2,\dots,N_\Tt.
\end{equation}
Hence, $\{\uv_n\}_{n=1}^{N_\Tt}$ represents the adjustable components in $\Fm_\At$. 
{For simplicity, the mutual coupling among the \ac{DMA} elements, which may distort the designed antenna-domain precoder when sufficiently strong, is neglected in this paper. Related studies on \ac{DMA} mutual coupling can be found in, e.g.,~\cite{Williams2022Electromagnetic,Prodhomme2024Mutual,Prodhomme2025Benefits}. Although often viewed as an impairment, mutual coupling can also be exploited to improve system performance when properly modeled~\cite{Prodhomme2025Benefits,Zheng2026MutualC}. Incorporating this effect into the proposed framework is left for future work.}

{Note that, although both $\Fm_\mathsf{RF}$ and $\Fm_\mathsf{A}$ provide beamforming capabilities, they play complementary rather than redundant roles. Specifically, $\Fm_\mathsf{RF}$ maps the outputs of a limited number of \ac{RF} chains to the inputs of the reconfigurable front end, whereas $\Fm_\mathsf{A}$ further shapes the resulting radiation through the reconfigurable antenna-domain responses. The hybrid use of both beamformers offers the following two benefits:
\begin{itemize}
	\item \emph{Balancing spectral and energy efficiency.} The \ac{RF}-domain precoder $\Fm_\mathsf{RF}$ is realized using \ac{RF} phase shifters, which provide direct phase control and strong beamforming capability at the expense of high power consumption (typically tens of $\mathrm{mW}$ per element). In contrast, the antenna-domain precoder $\Fm_\mathsf{A}$ is controlled by low-power semiconductor tuning components, such as varactors (typically $\mathrm{\mu W}$ or less per element), while the Lorentzian-constrained response limits beamforming flexibility. Their hybrid use therefore achieves a favorable balance between spectral and energy efficiency; see, e.g.,~\cite[Fig.~5]{Heath2025Tri}.
    \item \emph{Enhancing scalability.} Unlike conventional \ac{DMA}-based communication systems, where each \ac{DMA} typically requires a dedicated \ac{RF} chain~\cite{Shlezinger2019Dynamic,Zhang2022Beam,Kimaryo2023Downlink}, the analog precoder $\Fm_\mathsf{RF}$ maps signals from $N_\mathsf{RF}$ \ac{RF} chains to $N_\mathsf{T}$ \acp{DMA}. This architecture improves scalability by allowing more \acp{DMA} to be deployed without proportionally increasing the number of \ac{RF} chains.
\end{itemize}
}

\subsection{Signal Transmission and Performance Metric}
Consider a narrowband downlink signal model. Let $\Hm_k\in\mathbb{C}^{M_k\times N_\Ct N_\Tt}$ denote the downlink channel matrix from the \ac{BS} to the $k^\text{th}$ user. Under the tri-hybrid precoding architecture, the received signal at the $k^\text{th}$ user can be written as
\begin{align}\label{eq:yk}
	\yv_k = \Hm_k \Fm_\At \Fm_\mathsf{RF}\Fm_{\mathsf{BB},k}\sv_k +\! \sum_{i\neq k} \Hm_k\Fm_\At \Fm_\mathsf{RF}\Fm_{\mathsf{BB},i}\sv_i + \nv_k,
\end{align}
where $\Fm_{\mathsf{BB},k}\in\mathbb{C}^{N_\mathsf{RF}\times N_{\St,k}}$ is the digital precoder corresponding to $\sv_k$ such that $[\Fm_{\mathsf{BB},1},\Fm_{\mathsf{BB},2},\dots,\Fm_{\mathsf{BB},K}]=\Fm_\mathsf{BB}$. The first term in \eqref{eq:yk} is the desired signal, the second term represents the multiuser interference, and $\nv_k\sim\mathcal{CN}(\mathbf{0},\sigma_k^2\mathbf{I}_{M_k})$ models the additive white Gaussian noise at the $k^\text{th}$ user.

Assuming Gaussian signaling, the achievable spectral efficiency of the $k^\text{th}$ user is given by~\cite{Christensen2008Weighted, Sohrabi2016Hybrid}
\begin{multline}\label{eq:Rk}
	R_k = \log_2 \det \bigg(\mathbf{I}_{M_k}+\Hm_k \Fm_\mathsf{A}\Fm_\mathsf{RF}\Fm_{\mathsf{BB},k}\Fm_{\mathsf{BB},k}^\HH\Fm_\mathsf{RF}^\HH\Fm_\mathsf{A}^\HH\Hm_k^\HH\\
	\times\Big(\sum_{i\neq k}\Hm_k\Fm_\mathsf{A}\Fm_\mathsf{RF}\Fm_{\mathsf{BB},i}\Fm_{\mathsf{BB},i}^\HH\Fm_\mathsf{RF}^\HH\Fm_\mathsf{A}^\HH\Hm_k^\HH+\sigma_k^2\mathbf{I}_{M_k}\Big)^{-1}\bigg).
\end{multline}
To characterize the overall system performance, we adopt the \ac{WSR} defined as
\begin{equation}\label{eq:defR_new}
R = \sum_{k=1}^K \beta_k R_k,
\end{equation}
where $\beta_k> 0$ denotes the priority weight for the $k^\text{th}$ user.

\subsection{Problem Formulation}\label{sec:PF}
The goal of tri-hybrid beamforming is to jointly design the antenna-domain, RF-domain, and baseband precoders so as to maximize the system \ac{WSR}. This leads to the following optimization problem:
\begin{equation}\label{P0}
	\begin{aligned}
		\max_{\Fm_\mathsf{A},\Fm_\mathsf{RF},\Fm_\mathsf{BB}} \ & R \\
		\mathrm{s.t.} \ \quad & \big[\Fm_\mathsf{RF} \Fm_{\mathsf{BB}} \Fm_{\mathsf{BB}}^\HH \Fm_\mathsf{RF}^\HH \big]_{n,n} \leq P_n,\ \forall n,\\
		& |[\Fm_\mathsf{RF}]_{i,j}|^2 = 1/N_\Tt, \ \forall i,j,\\
		& |[\uv_n]_m|^2 = 1,\ \forall n,m.
	\end{aligned}
\end{equation}
{Different from existing studies on \ac{DMA} beamforming~\cite{Kimaryo2023Downlink,Zhang2022Beam}, fully connected \ac{DMA} architectures~\cite{Boateng2026Beamforming}, and single-user \ac{DMA}-based tri-hybrid systems~\cite{Castellanos2023Energy}, Problem~\eqref{P0} jointly accounts for multiuser \ac{MIMO} precoding, baseband and \ac{RF} processing, Lorentzian-constrained \ac{DMA} weights, and per-\ac{DMA} input power constraints. To the best of the authors’ knowledge, this particular problem formulation has not previously been investigated in the literature.} 
This problem is highly nonconvex due to the coupled optimization variables and the constant-modulus constraints imposed by the RF and antenna-domain hardware. 

To render the problem more tractable, we adopt a two-stage strategy commonly used in large-scale hybrid beamforming. Specifically, we first define an intermediate virtual fully digital precoder $\Fm_\mathsf{D} = \Fm_\mathsf{RF}\Fm_\mathsf{BB} \in \mathbb{C}^{N_\Tt \times N_\St}$ such that $\big[\Fm_\mathsf{D} \Fm_{\mathsf{D}}^\HH \big]_{n,n} \leq P_n,\ \forall n$. This allows us to relax~\eqref{P0} as
\begin{equation}\label{P1} \tag{$\mathcal{P}$1}
	\begin{aligned}
		\max_{\Fm_\mathsf{A},\Fm_\mathsf{D}} \quad & R \\
		\mathrm{s.t.} \ \quad & \big[\Fm_\mathsf{D} \Fm_{\mathsf{D}}^\HH \big]_{n,n} \leq P_n,\ \forall n,\\
		& |[\uv_n]_m|^2 = 1,\ \forall n,m.
	\end{aligned}
\end{equation}
Compared with \eqref{P0}, Problem~\ref{P1} reduces an optimization variable and removes the RF constant-modulus constraints, therefore admitting a simpler structure, while still respecting the per-DMA input power constraint and the physical constraints of the DMAs. After obtaining a solution $\{\Fm_\mathsf{A}^\mathsf{opt}, \Fm_\mathsf{D}^\mathsf{opt}\}$ to Problem~\ref{P1}, we recover the hybrid precoders by decomposing the fully digital precoder into the RF and baseband components. This leads to the following second-stage problem:
\begin{equation}\label{P2}\tag{$\mathcal{P}$2}
	\begin{aligned}
		\min_{\Fm_\mathsf{RF},\Fm_\mathsf{BB}} \quad & \|\Fm_\mathsf{D}^\mathsf{opt} - \Fm_\mathsf{RF}\Fm_\mathsf{BB}\|_\mathsf{F}^2\\
		\mathrm{s.t.} \ \ \quad & |[\Fm_\mathsf{RF}]_{i,j}|^2 = 1/N_\Tt, \ \forall i,j,\\
		& \big[\Fm_\mathsf{RF} \Fm_{\mathsf{BB}} \Fm_{\mathsf{BB}}^\HH \Fm_\mathsf{RF}^\HH \big]_{n,n} \leq P_n,\ \forall n.
	\end{aligned}	
\end{equation}

A large body of prior work has demonstrated that designing an unconstrained fully digital precoder followed by matrix decomposition yields near-optimal performance for large-scale hybrid architectures~\cite{Yu2016Alternating, Ni2017Near,Jin2018Hybrid}. This approach provides an effective balance between performance and tractability.

\section{Iterative Tri-Hybrid Beamforming}\label{solveP1}

This section develops a model-based iterative algorithm for solving Problem~\eqref{P0}, following the two-stage strategy presented in Section~\ref{sec:PF}. We primarily focus on Problem~\ref{P1}, which constitutes the main challenge of the tri-hybrid beamforming design, whereas Problem~\ref{P2} has been well studied in the literature and can be effectively handled using existing methods~\cite{Yu2016Alternating, Ni2017Near,Jin2018Hybrid}. Our solution to Problem~\ref{P1} builds on the \ac{WMMSE} equivalence detailed below.

\subsection{Preliminaries on WMMSE Equivalence}

Problem~\ref{P1} is non-convex due to the strong coupling among the optimization variables. A well-established approach to address this type of problem in multiuser MIMO systems is to exploit the equivalence between \ac{WSR} maximization and \ac{WMMSE} minimization~\cite{Christensen2008Weighted,Shi2011Iteratively,Zhao2023Rethinking,Pellaco2022Matrix}. Based on this relationship, we can reformulate Problem~\ref{P1} into an equivalent \ac{WMMSE} minimization problem, which enables a more tractable algorithmic design. The mathematical core of the \ac{WMMSE} equivalence is summarized in the following lemma adapted from~\cite{Zhao2023Rethinking}.

\begin{lemma}[\!\!\cite{Zhao2023Rethinking}] \label{lemma1}
	Given a matrix $\Am\in\mathbb{C}^{n\times\ell}$ and any positive definite matrix $\Nm\in\mathbb{C}^{n\times n}$, the following equality holds:
	\begin{equation}\label{eq:lemma1}
		\ln\det\big(\mathbf{I} + \Am\Am^\HH\Nm^{-1}\big) = 
		\max_{\substack{\Omegam \succ \mathbf{0},\,\Gammam}} \ln\det(\Omegam) - \mathrm{Tr}(\Omegam \Em(\Gammam)) + \ell,
	\end{equation}
	where  $\Omegam\in\mathbb{C}^{\ell\times \ell}$ and $\Gammam\in\mathbb{C}^{n\times \ell}$ are auxiliary variables, and $\Em(\Gammam)$ is a matrix function defined as
	\begin{equation}
		\Em(\Gammam) = (\mathbf{I} - \Gammam^\HH \Am)(\mathbf{I} - \Gammam^\HH \Am)^\HH 
		+ \Gammam^\HH \Nm \Gammam.
	\end{equation}
	The optimal $\Gammam^\mathsf{opt}$ and $\Omegam^\mathsf{opt}$ for the right-hand side of~\eqref{eq:lemma1} are given in closed-form as
	\begin{align}
		\Gammam^\mathsf{opt} &= (\Am \Am^\HH + \Nm)^{-1} \Am,\label{eq:Gamma_opt}\\	
		\Omegam^\mathsf{opt} &= \big(\Em(\Gammam^\mathsf{opt})\big)^{-1} = \big(\mathbf{I} - (\Gammam^\mathsf{opt})^\HH \Am\big)^{-1}.\label{eq:Omega_opt}
	\end{align}
\end{lemma}
\begin{proof}
See~\cite{Shi2011Iteratively, Zhao2023Rethinking}.
\end{proof}

Based on Lemma~1, we perform the following substitutions:
\begin{align}
	\Am &\leftarrow \Hm_k \Fm_\mathsf{A} \Fm_{\mathsf{D},k},\\
	\Nm &\leftarrow \sum_{i\neq k} \Hm_k \Fm_\mathsf{A} \Fm_{\mathsf{D},i} \Fm_{\mathsf{D},i}^\HH \Fm_\mathsf{A}^\HH \Hm_k^\HH + \sigma_k^2 \mathbf{I}_{M_k},
\end{align}
where $\Fm_{\mathsf{D},k}\in\mathbb{C}^{N_\Tt\times N_{\St,k}}$ is the virtual digital precoder corresponding to $\sv_k$ such that $[\Fm_{\mathsf{D},1},\Fm_{\mathsf{D},2},\dots,\Fm_{\mathsf{D},K}]=\Fm_\mathsf{D}$. Thus, Problem~\ref{P1} is equivalent to
\begin{equation}\label{MMSEM}
	\begin{aligned}
		\max_{\substack{
\Fm_\mathsf{A},\Fm_\mathsf{D} \\
{\{\Omegam_k \succ \mathbf{0}\},\{\Gammam_k\}}
}} & \sum_{k=1}^K \beta_k \Big(\log_2\det(\Omegam_k) - \mathrm{Tr}\big(\Omegam_k \Em_k(\Gammam_k)\big)\Big)\\
		\mathrm{s.t.} \ \ \quad & \big[\Fm_\mathsf{D} \Fm_{\mathsf{D}}^\HH \big]_{n,n} \leq P_n,\ \forall n,\\
		& |[\uv_n]_m|^2 = 1,\ \forall n,m.
	\end{aligned}
\end{equation}
Here, $\Em_k(\Gammam_k)$ is given by
\begin{multline}\label{eq:Ek}
	\Em_k(\Gammam_k) = (\mathbf{I} - \Gammam_k^\HH \Hm_k \Fm_\mathsf{A} \Fm_{\mathsf{D},k})(\mathbf{I} - \Gammam_k^\HH \Hm_k \Fm_\mathsf{A} \Fm_{\mathsf{D},k})^\HH\\
	+ \sum_{i\neq k} \Gammam_k^\HH \Hm_k \Fm_\mathsf{A} \Fm_{\mathsf{D},i} \Fm_{\mathsf{D},i}^\HH \Fm_\mathsf{A}^\HH \Hm_k^\HH \Gammam_k + \sigma_k^2 \Gammam_k^\HH \Gammam_k.
\end{multline}

While more variables are introduced, Problem~\eqref{MMSEM} becomes more tractable. The resulting problem admits a block-wise structure that naturally lends itself to an alternating optimization framework. Given a proper initialization, we can update variable blocks $\{\Omegam_k\}$, $\{\Gammam_k\}$, $\Fm_\mathsf{A}$, and $\Fm_\mathsf{D}$ sequentially while keeping the others fixed. In what follows, we show that each subproblem either admits a closed-form optimal solution or yields a closed-form update that guarantees monotonic improvement of the objective value.

\subsection{Update Auxiliary Variables $\{\Gammam_k\}$ and $\{\Omegam_k\}$}
According to~\eqref{eq:Gamma_opt}, we first update $\Gammam_k\in\mathbb{C}^{M_k\times N_{\St,k}}$ as
\begin{multline}\label{upd_Gammamk}
	\Gammam_k^\mathsf{opt} = \bigg(\sum_{i=1}^K \Hm_k \Fm_\mathsf{A} \Fm_{\mathsf{D},i}\Fm_{\mathsf{D},i}^\HH \Fm_\mathsf{A}^\HH \Hm_k^\HH + \sigma_k^2 \mathbf{I}_{M_k}\bigg)^{-1}\\
	\times \Hm_k \Fm_\mathsf{A} \Fm_{\mathsf{D},k},\ \forall k.
\end{multline}
With $\Gammam_k = \Gammam_k^\mathsf{opt}$, $\forall k$, we then update $\Omegam_k\in\mathbb{C}^{N_{\St,k}\times N_{\St,k}}$ according to~\eqref{eq:Omega_opt} as
\begin{equation}\label{upd_Omegak}
	\Omegam_k^\mathsf{opt} = \Big(\mathbf{I}_{N_{\St,k}}-\Gammam_k^\HH \Hm_k \Fm_\mathsf{A} \Fm_{\mathsf{D},k}\Big)^{-1},\ \forall k.
\end{equation}

\subsection{Update Precoding Variables $\Fm_\At$ and $\Fm_\Dt$} 
Following~\eqref{MMSEM}, we update $\Fm_\At$ and $\Fm_\Dt$ given updated $\{\Gammam_k, \Omegam_k\}$ by solving the following problem:
\begin{equation}\label{P3_full}
	\begin{aligned}
		\min_{\Fm_\At,\Fm_\Dt} \quad & \sum_{k=1}^K \beta_k \mathrm{Tr}\Big(\Omegam_k \Em_k\big(\Gammam_k\big)\Big)\\
		\mathrm{s.t.} \ \quad & \big[\Fm_\Dt \Fm_{\Dt}^\HH \big]_{n,n} \leq P_n,\ \forall n,\\
		& |[\uv_n]_m|^2 = 1,\ \forall n,m.
	\end{aligned}	
\end{equation}

Considering the per-\ac{DMA} input power constraint, we sequentially optimize the components associated with each \ac{DMA} to avoid coupled constraints. We partition $\Fm_\Dt$ as
\begin{equation}
	\begin{aligned}
		\Fm_\Dt &= [\Fm_{\Dt,1},\Fm_{\Dt,2},\dots,\Fm_{\Dt,K}] = \begin{bmatrix}
			\vv_1^\HH \\
			\vv_2^\HH \\
			\vdots \\
			\vv_{N_\Tt}^\HH
		\end{bmatrix},
	\end{aligned}
\end{equation}
where each $\vv_n^\HH\in\mathbb{C}^{1 \times N_{\St}}$ denotes the precoder for the signal fed into the $n^\text{th}$ \ac{DMA}. This partition ensures that each $\vv_n$ is subject to a single power constraint, i.e., $\|\vv_n\|_2^2 \leq P_n,\ \forall n$. 

Recall that $\Fm_\mathsf{A} = \mathrm{blkdiag}(\wv_{1},\wv_{2},\dots,\wv_{N_\Tt})$, where $\wv_n = (\Gm_n\uv_n + \mathsf{j}\,\etav_n)/2$ is the weight vector corresponding to the $n^\text{th}$ \ac{DMA} and $\uv_n$ is the tunable phase shift. Consequently, the digital precoder $\vv_n$ and the DMA tuning vector~$\uv_n$ naturally form a coupled variable block associated with the same physical unit. This motivates an inner block-wise optimization loop: we optimize $(\vv_n, \uv_n)$ for each \ac{DMA} sequentially, so that each subproblem involves only a single power constraint.

We now focus on optimizing $(\vv_n,\uv_n)$ associated with the $n^\text{th}$ \ac{DMA}. We partition $\Hm_k = [\Hm_{k,1},\Hm_{k,2},\dots,\Hm_{k,N_\Tt}]$, where each $\Hm_{k,n}\in\mathbb{C}^{M_k\times N_{\Ct}}$ is the channel matrix from the $n^\text{th}$ \ac{DMA} to the $k^\text{th}$ user. Moreover, we partition $\vv_n^\HH=[\vv_{n,1}^\HH,\vv_{n,2}^\HH,\dots,\vv_{n,K}^\HH]$, where $\vv_{n,k}^\HH\in\mathbb{C}^{1\times N_{\St,k}}$ is the precoder corresponding to symbol $\sv_k$ and the $n^\text{th}$ \ac{DMA}. By inspecting the objective function in~\eqref{P3_full} together with the expression of $\Em_k$ in~\eqref{eq:Ek}, we observe that the product $\Hm_k \Fm_\At \Fm_{\Dt,i}$ always appears as a composite term. Therefore, we can express $\Hm_k\Fm_\At \Fm_{\Dt,i}$ in terms of $\uv_n$ and $\vv_n$ as
\begin{equation}\label{eq:HFF}
	\begin{aligned}
		\Hm_k\Fm_\At \Fm_{\Dt,i} &= \sum_{n=1}^{N_\Tt} \Hm_{k,n} \wv_n \vv_{n,i}^\HH\\
		&= \Hm_{k,n} \wv_n \vv_{n,i}^\HH + \sum_{m\neq n} \Hm_{k,m} \wv_m \vv_{m,i}^\HH.
	\end{aligned}
\end{equation}
Substituting~\eqref{eq:HFF} into~\eqref{P3_full} and omitting all constant terms independent of $\uv_n$ and $\vv_n$, we obtain the subproblem for the $n^\text{th}$ \ac{DMA} as
\begin{equation}\label{eq:opt_idv}
	\begin{aligned}
		\min_{\uv_n,\vv_n} \quad & \|\vv_n\|_2^2\wv_n^\HH\Bm_{n,n}\wv_n + 2\mathrm{Re}\{\wv_n^\HH (\Qm_n-\Dm_n) \vv_n\}\\
		\mathrm{s.t.} \ \quad & \wv_n = (\Gm_n\uv_n + \mathsf{j}\,\etav_n)/2,\\
		& \|\vv_n\|_2^2 \leq P_n,\\
		& |[\uv_n]_m|^2 = 1,\ m=1,2,\dots,N_\Ct.
	\end{aligned}	
\end{equation}
where
\begin{align}
	\Bm_{q,p} &= \sum_{k=1}^K \beta_k \Hm_{k,q}^\HH \Gammam_k \Omegam_k \Gammam_k^\HH \Hm_{k,p}\in\mathbb{C}^{N_\Ct\times N_\Ct},\notag\\
	\Qm_n &= \sum_{m\neq n}\Bm_{n,m}\wv_m\vv_m^\HH\in\mathbb{C}^{N_\Ct\times N_\St}, \notag\\
	\Dm_n &= \begin{bmatrix}
		\beta_1 \Hm_{1,n}^\HH\Gammam_1 \Omegam_1,\dots,\beta_K \Hm_{K,n}^\HH\Gammam_K \Omegam_K \\
	\end{bmatrix}\in\mathbb{C}^{N_\Ct\times N_\St}.\notag
\end{align}

We further adopt an alternating optimization strategy for Problem~\eqref{eq:opt_idv} \ac{w.r.t.} $\uv_n$ and $\vv_n$, which results in simple closed-form update expressions, as detailed below.

\subsubsection{Update $\vv_n$ Given $\uv_n$} 
The optimization of $\vv_n$ given $\uv_n$ can be expressed as
\begin{equation}\label{eq:A1}
	\begin{aligned}
		\min_{\vv_n} \quad & a_n\|\vv_n\|_2^2 + 2\mathrm{Re}\{\dv_n^\HH \vv_n\}\\
		\mathrm{s.t.} \quad & \|\vv_n\|_2^2 \leq P_n,	
	\end{aligned}
\end{equation}
where $a_n = \wv_n^\HH \Bm_{n,n} \wv_n$ and $\dv_n = (\Qm_n^\HH - \Dm_n^\HH)\wv_n$. 
When $\dv_n\neq \mathbf{0}$, this subproblem has a closed-form solution as
\begin{equation}\label{eq:vnOptSolu}
	\vv_n^\mathsf{opt} = -\dv_n \min\bigg(\frac{1}{a_n}, \frac{\sqrt{P_n}}{\|\dv_n\|_2}\bigg).
\end{equation}
Since $\Omegam_k \succ \mathbf{0}$ by Lemma~\ref{lemma1}, $a_n > 0$ always holds. In the extreme case where $\dv_n=\mathbf{0}$, the optimal solution is $\vv_n^\mathsf{opt} = \mathbf{0}$.
\subsubsection{Update $\uv_n$ Given $\vv_n$}
Fixing $\vv_n$, the optimization of $\uv_n$ in~\eqref{eq:opt_idv} can be rewritten as 
\begin{equation}\label{eq:optun}
	\begin{aligned}
		\min_{\uv_n} \quad & \uv_n^\HH \Mm_n \uv_n + \mathrm{Re}\{\uv_n^\HH \bv_n\}\\
		\mathrm{s.t.} \quad & |[\uv_n]_m| = 1,\ \forall m,
	\end{aligned}	
\end{equation}
where $\Mm_n \!=\! \frac{\|\vv_n\|_2^2}{4}\Gm_n^\HH\Bm_{n,n} \Gm_n$ and $\bv_n\! =\! \Gm_n^\HH (\Qm_n\!-\! \Dm_n) \vv_n\! +\! 2\mathsf{j}\Mm_n\mathbf{1}$. Here, $\mathbf{1}$ denotes the $N_\Ct$-length all-ones vector. 

This subproblem is a unit-modulus quadratic program, which can be solved using techniques such as manifold optimization~\cite{boumal2023introduction}. However, such methods are iterative and may incur substantial computational overhead, especially since this is a subproblem appearing within an outer alternating-optimization loop. Computing the exact optimum of~\eqref{eq:optun} is practically unnecessary, and an appropriate update that ensures a monotonic non-increase of the objective is adequate. {Following the principle of majorization-minimization optimization~\cite{Sun2017Majorization,Breloy2021Majorization,He2022QCQP}, we instead minimize a quadratic upper bound of the objective, which leads to a computationally lightweight closed-form update given by} \begin{equation}\label{eq:un_upd}
\uv_n = e^{\mathsf{j}\arg\big((\rho_n\mathbf{I} - \Mm_n)\uv_n - \tfrac{1}{2}\bv_n\big)},
\end{equation}
where $\rho_n$ is chosen such that $\rho_n \geq \lambda_\mathsf{max}(\Mm_n)$, with $\lambda_\mathsf{max}(\Mm_n)$ denoting the maximum eigenvalue of $\Mm_n$. The proof of the monotonic non-increasing property of~\eqref{eq:un_upd} is provided in the Appendix.

\subsection{Algorithm Summary}

\begin{algorithm}[t]
	\caption{An Iterative Algorithm to Solve Problem~\ref{P1}}
	\label{ModelAlgo}
	\begin{algorithmic}[1]
		\Statex \textbf{Input:} Channel matrices $\{\Hm_k\}$, priority weights $\{\beta_k\}$, per-\ac{DMA} input power constraints $\{P_n\}$, complex waveguide response $\{\etav_n\}$, maximum iterations $L$.
		\State \textbf{Initialize:} $\Fm_\At^{(0)}$, $\Fm_\Dt^{(0)}$, set iteration index $\ell=0$.
		\Repeat
			\State Update $\big\{\Gammam_k^{(\ell+1)}\big\}$ using~\eqref{upd_Gammamk} with $\big\{\Fm_\At^{(\ell)}, \Fm_\Dt^{(\ell)}\big\}$.
			\State Update $\big\{\Omegam_k^{(\ell+1)}\big\}$ using~\eqref{upd_Omegak} with $\big\{\Fm_\At^{(\ell)}, \Fm_\Dt^{(\ell)}, \Gammam_k^{(\ell+1)}\!\big\}$.
			\For{$n=1$ to $N_\Tt$}
				\State Update $\vv_n^{(\ell+1)}$ using~\eqref{eq:vnOptSolu}.
				\State Update $\uv_n^{(\ell+1)}$ using~\eqref{eq:un_upd}.
			\EndFor
			\State $\wv_n^{(\ell+1)} = \big(\Gm_n\uv_n^{(\ell+1)} + \mathsf{j}\,\etav_n\big)/2$, $n=1,\dots,N_\Tt$
			\State $\Fm_\At^{(\ell+1)} = \mathrm{blkdiag}\big(\wv_{1}^{(\ell+1)},\dots,\wv_{N_\Tt}^{(\ell+1)}\big)$.
			\State $\Fm_\Dt^{(\ell+1)} = \big[\vv_1^{(\ell+1)},\dots,\vv_{N_\Tt}^{(\ell+1)}\big]^\HH$. 
			\State $\ell \leftarrow \ell + 1$
		\Until{$\ell \geq L$ or convergence.}
		
		\State Set $\Fm_\mathsf{A}^\mathsf{opt} = \Fm_\At^{(\ell)}$ and $\Fm_\mathsf{D}^\mathsf{opt} = \Fm_\Dt^{(\ell)}$.
	\end{algorithmic}
\end{algorithm}

Overall, the proposed iterative algorithm for solving~\ref{P1} is summarized in Alg.~\ref{ModelAlgo}. Although this algorithm follows an alternating optimization framework, each subproblem admits a closed-form update that guarantees a monotonic non-decrease of the objective value of~\ref{P1}. Since this objective function, i.e., \ac{WSR}, is upper-bounded, {Alg.~\ref{ModelAlgo} converges in terms of the objective value of Problem~\ref{P1}.}

After obtaining $\Fm_\mathsf{A}^\mathsf{opt}$ and $\Fm_\mathsf{D}^\mathsf{opt}$, the virtual digital precoder $\Fm_\mathsf{D}^\mathsf{opt}$ is decomposed into the hybrid precoders $\Fm_\mathsf{RF}^\mathsf{opt}$ and $\Fm_\mathsf{BB}^\mathsf{opt}$ by solving Problem~\ref{P2}. Since this decomposition problem has been extensively studied and numerous existing methods can be readily applied with minor adaptations, we do not detail the solution of Problem~\ref{P2} in this paper. Instead, we adopt the algorithm in~\cite{Jin2018Hybrid} whenever solving~\ref{P2} is required. {However, it should be noted that, since this RF-baseband decomposition is solved approximately, it does not guarantee local optimality for the original nonconvex Problem~\ref{P0}.}

\subsection{Complexity Analysis}\label{sec:CA_model}

The computational complexity of Alg.~\ref{ModelAlgo} is evaluated per outer iteration. Updating the auxiliary variables $\{\Gammam_k\}_{k=1}^K$ leads to complexity $\mathcal{O}\big(\sum_{k=1}^K M_k N_\Ct N_\Tt^2\big)$, while updating $\{\Omegam_k\}_{k=1}^K$ requires $\mathcal{O}\big(\sum_{k=1}^K N_{\St,k}^3 + N_{\St,k} N_\Ct N_\Tt^2\big)$. Constructing the coupling matrices ${\Bm_{q,p}}$, ${\Qm_n}$, and ${\Dm_n}$ dominates the matrix-assembly cost and scales as $\mathcal{O}\big(\sum_{k=1}^K N_\Tt^2 (N_\Ct M_k N_{\St,k} + N_\Ct^2 N_{\St,k})\big)$. The sequential DMA updates require eigenvalue and matrix-vector operations, yielding $\mathcal{O}(N_\Tt N_\Ct^3 + N_\Tt N_\Ct^2 N_\St)$. Therefore, the total per-iteration complexity can be written as
\begin{equation}\label{eq:Citer}
	\mathcal{C}_{\mathsf{iter}} = \mathcal{O}\Big(N_\Tt^2 N_\Ct^3 \sum_{k=1}^K M_k N_{\St,k}\Big).
\end{equation}
If the algorithm runs for $L$ iterations, the overall complexity is $\mathcal{O}(L\mathcal{C}_{\mathsf{iter}})$. For large DMA arrays, the dominant scaling is governed by the quadratic dependence on $N_\Tt$.

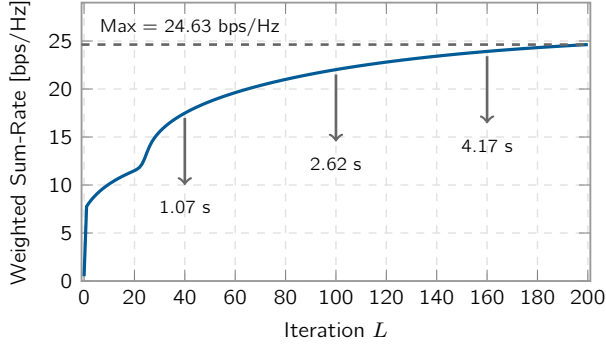
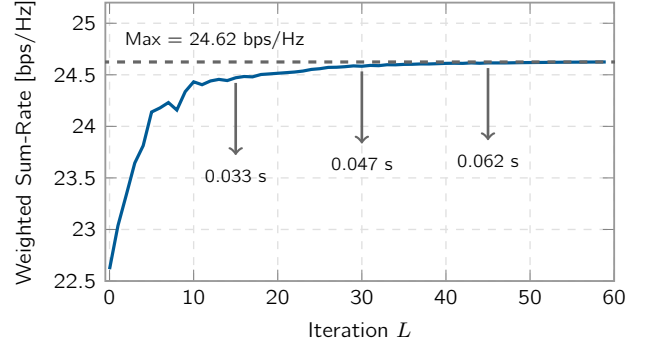
\begin{figure*}[t]
    \centering

    \subfloat[Solving Problem~\ref{P1} using the proposed Alg.~\ref{ModelAlgo}]{
        \input{figs/fig_P1.tex}
        \label{fig:a100}}
    \hfill
    \subfloat[Solving Problem~\ref{P2} using the algorithm in~\cite{Jin2018Hybrid}]{
        \input{figs/fig_P2.tex}
        \label{fig:b100}}

    \caption{Convergence performance and runtime of the algorithms for solving Problem~\ref{P1} and~\ref{P2}.}
    \label{fig:comparison}
\end{figure*}

\section{Deep Algorithm Unfolding}

{Although Alg.~\ref{ModelAlgo} has monotonic objective-value convergence and low per-iteration complexity, its multi-domain alternating structure may still require a large number of iterations. This limitation motivates the deep-unfolding design below, which treats Alg.~\ref{ModelAlgo} as a structured model-based backbone and learns to improve its finite-iteration performance from data.} 

\subsection{Rationale for the Unfolding Design}

Using a representative problem instance, Fig.~\ref{fig:comparison} evaluates the convergence performance of solving Problems~\ref{P1} and~\ref{P2} using the proposed Alg.~\ref{ModelAlgo} and~\cite{Jin2018Hybrid}, respectively. The system setup and simulation environment are described in Section~\ref{sec:SimSetUp}. It is clearly shown that solving Problem~\ref{P1} constitutes the main computational bottleneck. Specifically, Alg.~\ref{ModelAlgo} requires several seconds to converge, whereas solving Problem~\ref{P2} is approximately two orders of magnitude faster.

Motivated by these observations, we seek to accelerate Alg.~\ref{ModelAlgo} to solve Problem~\ref{P1} faster. We achieve this by incorporating a deep unfolding architecture. The goal is to learn a more efficient update rule from data, leading the algorithm to a solution with the desired performance in only a few iterations. Consider a real-time communication scenario in which only a limited number of algorithmic iterations (i.e., a constrained runtime) can be performed within each channel coherence interval. We assume that Alg.~\ref{ModelAlgo} is allowed to run for at most $L$ iterations. Under such constraints, asymptotic convergence becomes less relevant, and the performance achieved after $L$ iterations constitutes the primary performance metric. To optimize the $L$-iteration performance, we adopt an algorithm unfolding framework whose computational pipeline is depicted in Fig.~\ref{fig_unfold}. The unfolded algorithm consists of $L$ layers, each corresponding to one iteration of Alg.~\ref{ModelAlgo}. Building upon Alg.~\ref{ModelAlgo}, each layer is augmented with a set of learnable components. These components are detailed in the subsequent subsections.

{The proposed unfolded solver is related to, yet distinct from, existing \ac{WMMSE} algorithms and their variants. First, although Alg.~\ref{ModelAlgo} is developed within the \ac{WMMSE} framework~\cite{Christensen2008Weighted,Shi2011Iteratively,Zhao2023Rethinking}, its main contribution lies in deriving low-complexity updates for the resulting subproblems subject to per-\ac{DMA} input-power and Lorentzian constraints. Regarding the unfolding design, the existing \ac{GNN}-based unfolding method in~\cite{Chowdhury2021Unfolding} also incorporates learnable modules into a \ac{WMMSE}-type procedure. However, it focuses on scalar power allocation in single-antenna networks, whereas the present work targets complex-valued matrix precoding and antenna-domain \ac{DMA} tuning. Another relevant variant is the matrix-inverse-free \ac{WMMSE} algorithm in~\cite{Pellaco2022Matrix}, which modifies conventional \ac{WMMSE} updates to reduce the per-iteration complexity by avoiding matrix inversions. In contrast, our unfolding design addresses a different bottleneck by learning acceleration components that improve finite-iteration performance, thereby reducing the number of required iterations.}

\begin{figure*}[t]
  \centering
  \includegraphics[width=0.82\linewidth]{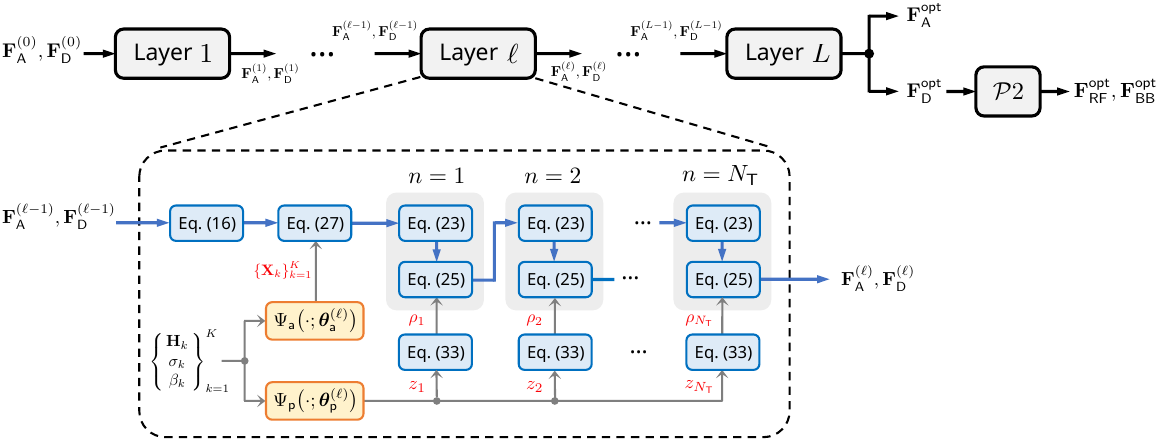}
\caption{Computational pipeline of the unfolded Alg.~\ref{ModelAlgo}.}
  \label{fig_unfold}
\end{figure*}

\subsection{Accelerating Updates of Auxiliary Variables}\label{sec:AUAV}
We first modify~\eqref{upd_Omegak} to accelerate the update of the auxiliary variables $\{\Gammam_k,\Omegam_k\}$ as 
\begin{equation}\label{eq:OmegaXk}
	\Omegam_k^\mathsf{opt} = \Big(\mathbf{I}_{N_{\St,k}}-\Gammam_k^\HH \Hm_k \Fm_\mathsf{A} \Fm_{\mathsf{D},k}\Big)^{-1} + \Xm_k + \Xm_k^\HH,\ \forall k,
\end{equation}
where $\Xm_k \in \mathbb{C}^{N_{\St,k} \times N_{\St,k}}$ is a trainable matrix that can be optimized to guide the update toward improved solutions beyond those obtained by the original update rule. Since the original $\Omegam_k^\mathsf{opt}$ is Hermitian, we add $\Xm_k + \Xm_k^\HH$ to preserve its Hermiticity. This augmented matrix is obtained by
\begin{equation}\label{eq:Psia}
	\{\Xm_k\}_{k=1}^K = \Psi_\mathsf{a}\big(\{\Hm_k,\sigma_k,\beta_k\}_{k=1}^K;\thetav_\mathsf{a}\big),
\end{equation}
where $\Psi_\mathsf{a}$ denotes a deep neural network and $\thetav_\mathsf{a}$ denotes its trainable weights. {The subscript ``$\mathsf{a}$'' is short for ``auxiliary,'' indicating that the corresponding network is used for auxiliary-variable updates.} We allocate a distinct network for each layer in Fig.~\ref{fig_unfold} with the same architecture but different weights, which is distinguished using a superscript $(\ell)$ on $\thetav_\mathsf{a}$. The design of $\Psi_\mathsf{a}$ is specified as follows.

\subsubsection{Deep Network Selection}
Function~\eqref{eq:Psia} indicates that the augmented terms $\{\Xm_k\}_{k=1}^K$ are learned from a given set of system parameters, i.e., a problem instance. In multiuser \ac{MIMO} systems, the number of active users $K$ is generally uncertain and may vary across scenarios, resulting in problem instances with different dimensions. Hence, $\Psi_\mathsf{a}$ must be trained in an inductive manner. In addition, this module should be permutation equivariant \ac{w.r.t.} user ordering, meaning that re-indexing the users should result in a consistent re-ordering of the outputs. To satisfy these requirements, we implement $\Psi_\mathsf{a}$ as an $I$-layer \ac{GraphSAGE} network~\cite{NIPS2017_5dd9db5e}. As illustrated in Fig.~\ref{fig:GNN1}, this is a \ac{GNN} defined on a fully connected graph with $K$ nodes, each representing a user. This architecture naturally supports a variable number of users and ensures permutation equivariance, since the learned message-passing mechanism is independent of both the graph size and the node indexing.

\begin{figure*}[t]
    \centering

    \subfloat[$\Psi_\mathsf{a}(\cdot;\thetav_\mathsf{a})$]{
        \includegraphics[width=0.45\linewidth]{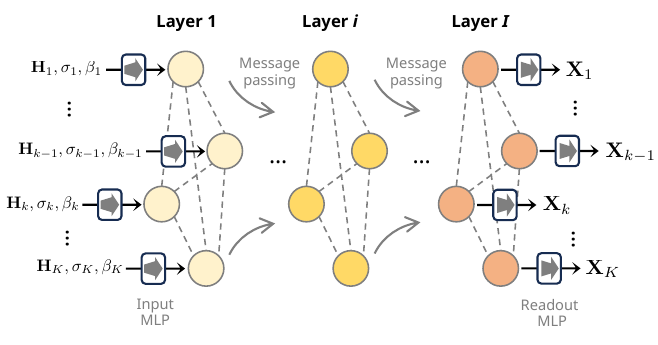}
        \label{fig:GNN1}}
    \hfill
    \subfloat[$\Psi_\mathsf{p}(\cdot;\thetav_\mathsf{p})$]{
        \includegraphics[width=0.495\linewidth]{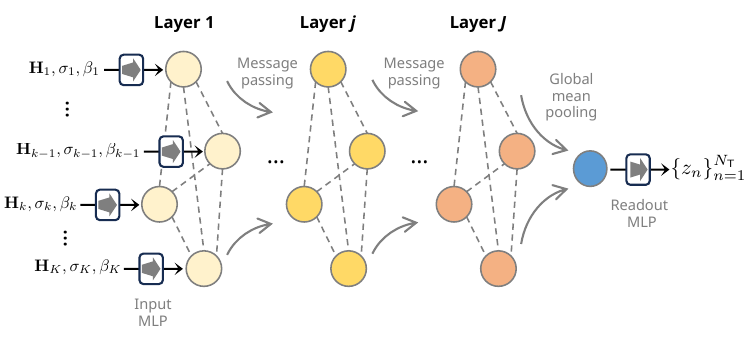}
        \label{fig:GNN2}}

    \caption{Architectures of the \ac{GNN} modules $\Psi_\mathsf{a}(\cdot;\thetav_\mathsf{a})$ and $\Psi_\mathsf{p}(\cdot;\thetav_\mathsf{p})$.}
    \label{fig:GNNs}
\end{figure*}

\subsubsection{Input Data Preprocessing}
The input of $\Psi_\mathsf{a}$ is the set of parameters that fully describes the problem instance. The input of the $k^\text{th}$ node  consists of the $k^\text{th}$ user-specific parameters $(\Hm_k, \sigma_k, \beta_k)$. We use a $K$-row matrix $\Ym^{(0)}$ to denote the initial node feature input to the first layer of the \ac{GNN}, where each row corresponds to the feature of a node/user. We illustrate in the following how $\Ym^{(0)}$ is computed for a given problem instance. 

A problem instance is uniquely determined by parameters $\{\Hm_k, \sigma_k,\beta_k\}_{k=1}^K$. Further, observing the definition of the spectral efficiency in~\eqref{eq:Rk}, we see that its landscape \ac{w.r.t.} $(\Fm_\At,\Fm_\mathsf{RF},\Fm_\mathsf{BB})$ does not change if we scale $\{\Hm_k, \sigma_k \}_{k=1}^K$. More precisely, we have 
\begin{multline}
	R_k\big(\Fm_\At,\Fm_\mathsf{RF},\Fm_\mathsf{BB};\{\Hm_k, \sigma_k \}_{k=1}^K\big) \\= R_k\big(\Fm_\At,\Fm_\mathsf{RF},\Fm_\mathsf{BB};\{a\Hm_k, a\sigma_k \}_{k=1}^K\big),
\end{multline}
given any non-zero number $a$. Similarly, scaling the user weights $\{\beta_k\}_{k=1}^K$ only introduces a common multiplicative factor to the landscape of \ac{WSR} \ac{w.r.t.} $(\Fm_\At,\Fm_\mathsf{RF},\Fm_\mathsf{BB})$. 

To mitigate these ambiguities and also ensure stable model training, we first normalize the system parameters. Specifically, we define
\begin{align}
	\Hm & \triangleq \begin{bmatrix}
		\sigma_1 & \mathrm{vec}(\mathrm{Re}\{\Hm_1\})^\TT & \mathrm{vec}(\mathrm{Im}\{\Hm_1\})^\TT \\
		\vdots & \vdots & \vdots \\
		\sigma_K & \mathrm{vec}(\mathrm{Re}\{\Hm_K\})^\TT & \mathrm{vec}(\mathrm{Im}\{\Hm_K\})^\TT
	\end{bmatrix},\notag\\
	\betav &\triangleq \begin{bmatrix}
		\beta_1,\ \beta_2,\ \dots,\ \beta_K
	\end{bmatrix}^\TT.\notag
\end{align}
Note that the construction of $\mathbf{H}$ requires the assumption $M_1 = \cdots = M_K \triangleq M$, which can be freely specified as a design parameter at the \ac{BS} for downlink beamforming. If this condition is not satisfied, the channel matrices $\{\mathbf{H}_k\}_{k=1}^K$ can be appropriately truncated or cyclically extended to ensure dimensional alignment. Then, we normalize them as 
\begin{align}\label{eq:Htilde}
	\tilde{\Hm} = \frac{\sqrt{K(2MN_\Ct N_\Tt+1)}\Hm}{\|\Hm\|_\Ft},\quad \tilde{\betav} = \frac{\sqrt{K}\betav}{\|\betav\|_2},
\end{align}
so that the squared norms of $\tilde{\Hm}$ and $\tilde{\betav}$ equal their respective dimensions. 
A simple construction of $\Ym^{(0)}$ is thus given by $[\tilde{\Hm}, \tilde{\betav}]$. When using equal user weights, i.e., $\beta_1=\cdots=\beta_K$, the vector $\tilde{\betav}$ can be omitted. To further enhance the representational capacity of the deep model, an \ac{MLP} can be employed to extract latent features prior to the \ac{GNN}, yielding
\begin{equation}\label{eq:Y0}
	\Ym^{(0)} = \mathrm{MLP}\big([\tilde{\Hm},\ \tilde{\betav}];\thetav_\mathsf{a}^{\mathsf{in}}\big),
\end{equation}
where $\thetav_\mathsf{a}^{\mathsf{in}}$ denotes the trainable weights of this \ac{MLP}.

\subsubsection{Message Passing}
Let \(\yv_k^{(i)}\) denote the feature of node $k$ at layer $i$. The corresponding node feature matrix of this layer is then defined as \(\Ym^{(i)} \triangleq \big[\yv_1^{(i)}, \yv_2^{(i)}, \dots, \yv_K^{(i)}\big]^\TT\).
The message passing from layer $i$ to layer $i+1$ using the \ac{GraphSAGE} algorithm can be described as~\cite{NIPS2017_5dd9db5e}
\begin{equation}\notag
	\yv_k^{(i+1)} = \phi\bigg(\Wm^{(i)}\bigg[\big(\yv_k^{(i)}\big)^\TT, \Big(\frac{1}{K-1} \sum_{r\neq k}\yv_r^{(i)}\Big)^\TT\bigg]^\TT\bigg),\ \forall k
\end{equation}
where $\Wm^{(i)}$ denotes the trainable weight associated with layer~$i$, and $\phi(\cdot)$ denotes a nonlinear activation function, such as the rectified linear unit (ReLU).

\subsubsection{Final Feature Readout}
We read out the matrix $\Xm_k$ from node $k$ in the last layer of the \ac{GNN} through another \ac{MLP}. Mathematically, we have
\begin{equation}
	\big[\mathrm{vec}(\mathrm{Re}(\Xm_k))^\TT,\! \mathrm{vec}(\mathrm{Im}(\Xm_k))^\TT\big]^\TT = \mathrm{MLP}\big(\yv_k^{(I)};\thetav_\mathsf{a}^\mathsf{out}\big),
\end{equation}
where $\thetav_\mathsf{a}^\mathsf{out}$ denotes the trainable weights of this \ac{MLP}. 

Based on the above architecture, the set of trainable parameters in module $\Psi_\mathsf{a}$ is given by $\thetav_\mathsf{a}=\big\{\thetav_\mathsf{a}^\mathsf{in},\{\Wm^{(i)}\}_{i=1}^{I},\thetav_\mathsf{a}^\mathsf{out}\big\}$.

\subsection{Accelerating Updates of Precoders}

In addition to the auxiliary variables, we further accelerate the updates of the precoder-related variables $\{\vv_n,\uv_n\}$ by unfolding the iterative rule in~\eqref{eq:un_upd} into a learnable architecture. Rather than introducing an augmented term into~\eqref{eq:un_upd} as in~\eqref{eq:OmegaXk}, we adopt a more restrained and computationally efficient strategy, i.e., using a neural network to adaptively select the hyperparameter $\rho_n$. Because this module outputs only scalar parameters, it remains lightweight. Specifically, we parameterize $\rho_n$ in~\eqref{eq:un_upd} as
\begin{equation}\label{eq:rhoz}
\rho_n = \lambda_{\mathsf{max}}(\Mm_n) + z_n, \quad \forall n,
\end{equation}
where the offsets $\{z_n\}_{n=1}^{N_\Tt}$ are generated by a deep network
\begin{equation}
\{z_n\}_{n=1}^{N_\Tt} = \Psi_\mathsf{p}\big(\{\Hm_k, \sigma_k, \beta_k\}_{k=1}^K; \thetav_\mathsf{p}\big).
\end{equation}
{Similarly, the subscript ``$\mathsf{p}$'' here is short for ``precoder,'' indicating that the corresponding network is used for precoder updates.}
Again, each layer in Fig.~\ref{fig_unfold} has an independent $\Psi_\mathsf{p}$ module with its own set of weights $\thetav_\mathsf{p}^{(\ell)}$.

The function $\Psi_\mathsf{p}$ is implemented as a $J$-layer \ac{GNN}, whose architecture is illustrated in Fig.~\ref{fig:GNN2}. The input data preprocessing and message-passing procedure in $\Psi_\mathsf{p}$ follow the same design as described in Section~\ref{sec:AUAV}. Since the desired outputs $\{z_n\}_{n=1}^{N_\Tt}$ are global quantities that are independent of the number of nodes/users, we employ a graph-level mean pooling operator in the final layer to aggregate node-wise features into a global representation. Specifically, letting $\Ym^{(J)}$ denote the node feature matrix of the last GNN layer, we compute
\begin{equation}
\yv^\TT = \frac{1}{K}\mathbf{1}^\TT \Ym^{(J)}.
\end{equation}
The desired output $\zv = [z_1,\dots,z_{N_\Tt}]^\TT$ is then obtained via an \ac{MLP} readout layer, i.e.,
\begin{equation}
	\zv = \mathrm{MLP}(\yv; \thetav_\mathsf{p}^\mathsf{out}).
\end{equation}

\subsection{Model Training}
The proposed unfolded algorithm becomes effective only after proper training. A few key design aspects of the training process are discussed below.

\subsubsection{Loss Function}\label{sec:LF}

Let $\mathcal{H}\triangleq\{\Hm_k, \sigma_k, \beta_k\}_{k=1}^K$ denote a problem instance for brevity. We denote the mapping implemented by each layer in Fig.~\ref{fig_unfold} as $g(\cdot)$, such that
\begin{equation}
	\big\{\Fm_\At^{(\ell)},\Fm_\Dt^{(\ell)}\big\}
	= g\big(\Fm_\At^{(\ell-1)}, \Fm_\Dt^{(\ell-1)}; \mathcal{H}, \thetav^{(\ell)}\big),
\end{equation}
where $\thetav^{(\ell)}$ denotes the set of trainable parameters at layer $\ell$, consisting of $\thetav_\mathsf{a}^{(\ell)}$ and $\thetav_\mathsf{p}^{(\ell)}$. Stacking $L$ layers together defines the overall mapping $G(\cdot)$ as:
\begin{equation}
	\big\{\Fm_\At^{(L)}, \Fm_\Dt^{(L)}\big\}
	= G\big(\Fm_\At^{(0)}, \Fm_\Dt^{(0)}; \mathcal{H}, \Thetam\big),
\end{equation}
where $\Thetam \triangleq \{\thetav^{(1)}, \dots, \thetav^{(L)}\}$ collects all trainable parameters.

Our objective is to maximize the system sum-rate achieved after $L$ layers. Since the sum-rate is a function of both the problem instance $\mathcal{H}$ and the resulting precoders, the achieved performance can be written as
\begin{equation}
	R\big(\Fm_\At^{(L)}, \Fm_\Dt^{(L)}; \mathcal{H}\big)
	= R\Big(G\big(\Fm_\At^{(0)}, \Fm_\Dt^{(0)}; \mathcal{H}, \Thetam\big); \mathcal{H}\Big).
\end{equation}
Given a dataset of problem instances $\bm{\mathcal{D}}$, we define the training loss as
\begin{equation}\label{eq:LTheta}
	\mathcal{L}(\Thetam)
	= -\mathbb{E}_{\mathcal{H}\sim\bm{\mathcal{D}}}\ 
	R\Big(G\big(\Fm_\At^{(0)}, \Fm_\Dt^{(0)}; \mathcal{H}, \Thetam\big); \mathcal{H}\Big).
\end{equation}

Model training amounts to minimizing $\mathcal{L}(\Thetam)$ \ac{w.r.t.} $\Thetam$, which is a nonconvex optimization problem. Nevertheless, since all components in~\eqref{eq:LTheta} are differentiable, the optimization can be efficiently performed using stochastic gradient descent. This is an unsupervised learning method; only problem instances are required, without any labeled optimal solutions, which facilitates practical implementation.

\subsubsection{Iteration Initialization}

A critical aspect of the proposed framework is that the initialization of the $L$-layer iterations must be fixed. This can be seen from~\eqref{eq:LTheta}, where the loss explicitly depends on the initial precoders $\big\{\Fm_\At^{(0)}, \Fm_\Dt^{(0)}\big\}$. Any change in the initialization alters the optimization landscape of the training loss. Therefore, the initialization must remain identical during both the training and inference phases. A similar observation was reported in~\cite{Sun2018Learning}.

\subsubsection{\texorpdfstring{Model Weight Initialization}{Model Weight Initialization}}\label{sec:MWI}

The model parameters $\Thetam$ can be initialized in various ways. We recommend initializing the trainable weights in the output layer of each module to zero, while initializing all other weights using standard schemes such as Xavier or Kaiming initialization. This strategy significantly improves training stability. As can be observed from~\eqref{eq:OmegaXk} and~\eqref{eq:rhoz}, this initialization ensures that the initial augmented terms $\{\Xm_k\}$ and $\{z_n\}$ are zero. Consequently, before training, the unfolded architecture exactly reproduces the original Alg.~\ref{ModelAlgo}. This provides a strong and meaningful starting point, allowing training to progressively improve upon the model-based solution using data. In contrast, a fully random initialization of $\Thetam$ may lead to poor initial performance, which is often worse than the model-based baseline, thereby slowing training convergence.

{
\subsection{Complexity and Scalability Discussion}\label{sec:complexity_scalability}

\subsubsection{GNN Complexity}

We first quantify the cost of the trainable \ac{GNN} modules. Following the standard \ac{GraphSAGE} message-passing implementation~\cite{NIPS2017_5dd9db5e}, let $\mathcal{E}$ denote the set of graph edges used for neighbor aggregation, where $|\mathcal{E}|=\mathcal{O}(K^2)$ for the fully connected graph used in this paper. Let $M_1=\dots=M_K=M$ and omit $\tilde{\betav}$. The input feature dimension of each user node is $d = 2M N_\Ct N_\Tt + 1$. For a generic module with $S$ \ac{GraphSAGE} layers, hidden width $h$, and output dimension $r$, its forward complexity can be summarized as
\begin{equation}
	\mathcal{C}_{\mathsf{GS}}(d,h,r,S)
	= \mathcal{O}\big(K\mathcal{C}_{\mathsf{in}}(d,h)
	+ S(|\mathcal{E}|h+Kh^2)
	+ \xi\mathcal{C}_{\mathsf{out}}(h,r)\big), \notag
\end{equation}
where $\mathcal{C}_{\mathsf{in}}(d,h)$ and $\mathcal{C}_{\mathsf{out}}(h,r)$ are the input and readout \ac{MLP} costs, respectively, and $\xi=K$ for node-level readout and $\xi=1$ for graph-level readout after mean pooling. In each unfolded layer, assuming $N_{\St,1}=\dots=N_{\St,K}=N_{\St,\mathsf{user}}$, $\Psi_\mathsf{a}$ is node-level with output dimension $r_\mathsf{a}=2N_{\St,\mathsf{user}}^2$, whereas $\Psi_\mathsf{p}$ is graph-level with output dimension $r_\mathsf{p}=N_\Tt$. Hence, the additional \ac{GNN} cost per unfolded layer is
\begin{equation}
	\mathcal{C}_{\mathsf{GNN}}
	= \mathcal{C}_{\mathsf{GS}}(d,h_\mathsf{a},r_\mathsf{a},S_\mathsf{a})
	+ \mathcal{C}_{\mathsf{GS}}(d,h_\mathsf{p},r_\mathsf{p},S_\mathsf{p}).
\end{equation}
Thus, the \ac{GNN} cost grows linearly with $d=\mathcal{O}(MN_\Ct N_\Tt)$ through the input \acp{MLP}, quadratically with $K$ through the fully connected user graph, linearly with the number of \ac{GraphSAGE} layers, and polynomially with the hidden widths through the linear transformations and readout.

\subsubsection{Offline Training Complexity}

Let $|\boldsymbol{\mathcal{D}}|$ denote the number of training problem instances, $E$ the number of training epochs, $B$ the batch size, and $P_\Theta$ the number of trainable parameters. For each problem instance, a forward pass runs $L$ unfolded layers, where each layer consists of the model-based update with cost $\mathcal{C}_{\mathsf{iter}}$ evaluated in~\eqref{eq:Citer} and the trainable \ac{GNN} modules with cost $\mathcal{C}_{\mathsf{GNN}}$. The \ac{WSR}-based loss in~\eqref{eq:LTheta} is then backpropagated through the unfolded computation graph. Under standard reverse-mode automatic differentiation, as implemented in modern deep-learning platforms such as PyTorch~\cite{NEURIPS2019_bdbca288}, the backward pass has the same asymptotic order as the forward pass up to a constant factor. Hence, the offline training complexity can be expressed as
\begin{equation}
	\mathcal{O}\!\left(
	E|\boldsymbol{\mathcal{D}}|\kappa L
	(\mathcal{C}_{\mathsf{iter}}+\mathcal{C}_{\mathsf{GNN}})
	+ E\left\lceil\frac{|\boldsymbol{\mathcal{D}}|}{B}\right\rceil P_\Theta
	\right),
\end{equation}
where $\kappa$ is a constant accounting for the backward pass, and the second term represents the per-batch cost of updating the trainable network parameters, which scales linearly with $P_\Theta$ for standard optimizers such as SGD or Adam. 

\subsubsection{Online Inference Complexity}

During online inference, the network parameters are fixed and no gradient computation or optimizer update is performed. For one problem instance, the real-time cost is therefore
\begin{equation}
	\mathcal{O}\!\left(L(\mathcal{C}_{\mathsf{iter}}+\mathcal{C}_{\mathsf{GNN}})
	+\mathcal{C}_{\mathcal{P}2}\right),
\end{equation}
where $\mathcal{C}_{\mathcal{P}2}$ is the cost of the final \ac{RF}-baseband decomposition in Problem~\ref{P2}. The computational complexity of solving this decomposition problem has been extensively evaluated in the literature; an example can be found in~\cite{Jin2018Hybrid}. Compared with Alg.~\ref{ModelAlgo}, the unfolded solver adds the \ac{GNN} inference cost $\mathcal{C}_{\mathsf{GNN}}$ in each unfolded layer, but it typically uses a much smaller number of layers $L$. The reduction in the number of required iterations is the main source of the runtime savings.

\subsubsection{Additional Scalability Remarks}

Two additional scalability aspects are worth noting. First, the number of \ac{RF} chains $N_\mathsf{RF}$ does not enter the dimensions of $\Psi_\mathsf{a}$ or $\Psi_\mathsf{p}$, because the unfolded solver operates on the virtual fully digital precoder $\Fm_\Dt$ in Problem~\ref{P1}. Thus, changing $N_\mathsf{RF}$ only affects the final RF-baseband decomposition cost $\mathcal{C}_{\mathcal{P}2}$ and the achievable decomposition accuracy. Second, the \ac{GraphSAGE}-based architecture is inductive with respect to the user graph. Although changing $K$ modifies the number of graph nodes and edges, and thus the message-passing complexity, it does not alter the network architecture or require retraining.
}

\section{Numerical Results}

This section presents the simulation results. The simulation setup is described first, followed by the evaluation of the proposed method.

\subsection{Simulation Setup}\label{sec:SimSetUp}

\subsubsection{System Configuration}

\begin{table}[t]
\centering
\caption{Simulation System Parameters}
\label{tab:system_parameters}

\begin{threeparttable}
\begin{tabular}{l c}
\toprule
\textbf{Parameter} & \textbf{Value} \\
\midrule
Number of users $K$ & $3$--$5$ \\
Number of antennas per user $M_k$ & $4$ \\
Number of data streams per user $N_{\mathsf{S},k}$\qquad & $2$ \\
Number of RF chains $N_{\mathsf{RF}}$ & $10$ \\
Number of DMAs $N_{\mathsf{T}}$ & $20$ \\
Number of elements per DMA $N_{\mathsf{C}}$ & $5$ \\
Per-DMA power constraint $P_n$ & $\unit[1]{mW}$ \\
Carrier frequency & $\unit[28]{GHz}$ \\
Bandwidth & $\unit[20]{MHz}^\sharp$ \\
Noise power spectral density & $\unit[-174]{dBm/Hz}$ \\
\bottomrule
\end{tabular}

\begin{tablenotes}
\footnotesize
\item[$\sharp$] The mmWave band at $\unit[28]{GHz}$ can support wider bandwidths up to hundreds of MHz. Here, a $\unit[20]{MHz}$ bandwidth is adopted to be consistent with the narrowband assumption.
\end{tablenotes}

\end{threeparttable}
\end{table}

Unless otherwise specified, the system parameters are set as listed in Table~\ref{tab:system_parameters}. The elements in each \ac{DMA} are assumed to be linearly and uniformly distributed with half-wavelength spacing. The signal propagation inside the \ac{DMA} (i.e., the complex waveguide response) of the $m^\text{th}$ element of the $n^\text{th}$ \ac{DMA} is modeled as~\cite{Zhang2022Beam}
\begin{equation}
\eta_{n,m} = e^{-d_{n,m}(\varrho_m + \mathsf{j}\,\varpi_m)},
\end{equation}
where $d_{n,m}$ denotes the one-dimensional location of the $m^\text{th}$ element in the $n^\text{th}$ \ac{DMA}, with the feeding point taken as the origin. Here, $\varrho_m$ represents the waveguide attenuation coefficient and $\varpi_m$ denotes the signal wavenumber. Following~\cite{Zhang2022Beam}, we set $\varrho_m = \unit[0.6]{m^{-1}}$.

\subsubsection{Training Settings}\label{sec:TS}
The proposed \ac{GNN}-based unfolding network is implemented using the PyTorch platform~\cite{NEURIPS2019_bdbca288}. Throughout the evaluation, equal user weights are adopted, i.e., $\beta_1=\dots=\beta_K=1$. Consequently, the input feature of each user node in the \acp{GNN} $\Psi_\mathsf{a}$ and $\Psi_\mathsf{p}$ corresponds to a row of $\tilde{\Hm}$ defined in~\eqref{eq:Htilde}, which has a dimension of $2M_k N_\mathsf{C} N_\mathsf{T} + 1 = 801$. These features are first processed by a shared input \ac{MLP} to extract latent representations, then propagated through the multi-layer \ac{GraphSAGE} network, and finally mapped to the output via a shared readout \ac{MLP}, as illustrated in Fig.~\ref{fig:GNNs}. The number of layers and the width of each layer in the deep networks are summarized in Table~\ref{tab:network_architecture}. All networks employ the ReLU activation function. 

\begin{table}[t]
    \centering
    \caption{Network architecture of unfolding modules $\Psi_{\mathsf a}$ and $\Psi_{\mathsf p}$}
    \label{tab:network_architecture}
    \begin{tabular}{cccc}
        \toprule
        \textbf{Module} 
        & \textbf{Input MLP} 
        & \textbf{GNN layers} 
        & \textbf{Readout MLP} \\
        \midrule
        $\Psi_{\mathsf a}$ 
        & $801 \rightarrow 406 \rightarrow 12$ 
        & $12 \rightarrow 12 \rightarrow 12$ 
        & $12 \rightarrow 10 \rightarrow 8$ \\
        $\Psi_{\mathsf p}$ 
        & $801 \rightarrow 415 \rightarrow 30$ 
        & $30 \rightarrow 30 \rightarrow 30$ 
        & $30 \rightarrow 25 \rightarrow 20$ \\
        \bottomrule
    \end{tabular}
\end{table}

We adopt the DeepMIMO dataset~\cite{Alkhateeb2019} to train the proposed unfolding network. As an unsupervised learning approach, the training relies on a collection of unlabeled problem instances. As defined in Section~\ref{sec:LF}, each problem instance $\mathcal{H}$ comprises the channel matrices and noise levels of $K$ users. To enhance generalization, we allow the number of users to vary across instances, where $K$ is randomly generated within the range $3 \leq K \leq 5$, as also shown in Table~\ref{tab:system_parameters}. Specifically, we collect a total of 28,595 channel realizations from the New York, Los Angeles, and Chicago scenarios provided by DeepMIMO. These channel data are then randomly combined to construct 20,000 problem instances, forming the training dataset $\boldsymbol{\mathcal{D}}$. For evaluation, we use an independent test set consisting of 1,000 problem instances generated from the Miami scenario. This separation between training and testing data completely avoids data leakage and ensures a fair performance assessment. The network is trained using the Adam optimizer~\cite{Kingma2015Adam}, with a learning rate progressively decayed from $10^{-3}$ to $10^{-6}$. We further employ dropout~\cite{Srivastava2014Dropout} to enhance generalization capability, with a dropout rate of $10\%$.

\subsubsection{Evaluation Settings}
All runtime evaluations are conducted on a MacBook Pro (macOS 26.3) with an Apple M1 Pro chip. To better position the proposed method within the existing literature, we adopt two benchmark algorithms based on prior work. As the specific multiuser \ac{MIMO} \ac{DMA}-based tri-hybrid \ac{WSR} problem in~\eqref{P0} has not been directly studied in prior work, no off-the-shelf benchmark is available for exactly the same setting. Therefore, the adopted benchmarks are constructed by combining several existing methods developed for related problems, as detailed below:
\begin{itemize}
	\item \textbf{AO-MO:} This is an algorithm originally proposed in~\cite{Kimaryo2023Downlink} to address the multiuser downlink beamforming problem with \acp{DMA}. It adopts an alternating optimization framework that iteratively applies the \ac{WMMSE} and Riemannian manifold optimization algorithms to optimize $\Fm_\At$ and $\Fm_\Dt$. Since the original method in~\cite{Kimaryo2023Downlink} only designs $\Fm_\At$ and $\Fm_\Dt$, we introduce an additional step that decomposes the resulting fully digital precoder $\Fm_\Dt$ into $\Fm_{\mathsf{RF}}$ and $\Fm_{\mathsf{BB}}$ using the algorithm in~\cite{Jin2018Hybrid}, thereby adapting it to our tri-hybrid beamforming setting. The manifold optimization steps are implemented using the Pymanopt toolbox~\cite{JMLR:v17:16-177}. 
	\item \textbf{Mat. Decomp.:} Another approach is based on matrix decomposition, which, as discussed in Section~\ref{sec:PF}, is widely used in conventional hybrid digital-analog beamforming~\cite{Yu2016Alternating,Ni2017Near,Jin2018Hybrid} due to its reduced complexity. To the best of our knowledge, this framework has not been extended to tri-hybrid \ac{MIMO}. Here, we adapt it to the tri-hybrid setting and employ it as a second benchmark. Specifically, we first optimize a virtual precoder $\tilde{\Fm} = \Fm_\At \Fm_\mathsf{RF} \Fm_\mathsf{BB}$  subject only to a total power constraint $\|\tilde{\Fm}\|_\Ft^2 \leq \sum_{n=1}^{N_\Tt} P_n$, which can be efficiently solved using, e.g., the algorithm in~\cite{Zhao2023Rethinking}. Let $\tilde{\Fm}^\mathsf{opt}$ denote the resulting solution. We then decompose $\tilde{\Fm}^\mathsf{opt}$ into $\Fm_\At$ and $\Fm_\Dt$ by solving
		\begin{equation}
		\begin{aligned}
		\min_{\Fm_\At,\Fm_\Dt}\ & \|\tilde{\Fm}^\mathsf{opt}-\Fm_\At\Fm_\Dt\|_\Ft^2\\
		\mathrm{s.t.}\ & \big[\Fm_\mathsf{D} \Fm_{\mathsf{D}}^\HH \big]_{n,n} \leq P_n,\ \forall n,\\
		& |[\uv_n]_m|^2 = 1,\ \forall n,m,
		\end{aligned}
		\end{equation}
		which can be efficiently handled via alternating optimization on Riemannian manifolds~\cite{JMLR:v17:16-177}, but with a substantially simpler objective function. Finally, $\Fm_\Dt$ is factorized into $\{\Fm_\mathsf{RF},\Fm_\mathsf{BB}\}$ by solving Problem~\ref{P2} using the method presented in~\cite{Jin2018Hybrid}.
\end{itemize}

{
\subsection{Model-Based Algorithm Performance}

Before evaluating the unfolded solver, this subsection first assesses the performance of the developed Alg.~\ref{ModelAlgo}. Figure~\ref{fig_twoStage} is presented to quantify the performance loss introduced by the second-stage RF-baseband decomposition. The two-stage curves evaluate the realizable precoders obtained after approximating $\Fm_\Dt$ by $\Fm_\mathsf{RF}\Fm_\mathsf{BB}$ for different numbers of RF chains, while the virtual fully-digital curve serves as a reference. For most input-power levels and numbers of \ac{RF} chains, we observe that the adopted RF-baseband decomposition introduces only a negligible performance loss compared to the virtual fully-digital reference, indicating its effectiveness. Moreover, increasing $N_\mathsf{RF}$ enlarges the feasible representation space of the RF-baseband decomposition. In the low-power regime, the two two-stage curves ($N_\mathsf{RF}=8$ and $N_\mathsf{RF}=6$) nearly overlap. In the high-power regime, however, the impact of the decomposition becomes more pronounced. The $N_\mathsf{RF}=8$ curve remains close to the fully-digital reference, whereas the $N_\mathsf{RF}=6$ curve exhibits an increasing gap as the input-power budget grows. This behavior indicates that the RF-baseband decomposition error can become a dominant limiting factor when the RF-chain dimension is restrictive and the per-\ac{DMA} input-power constraint is high.

\begin{figure}[t]
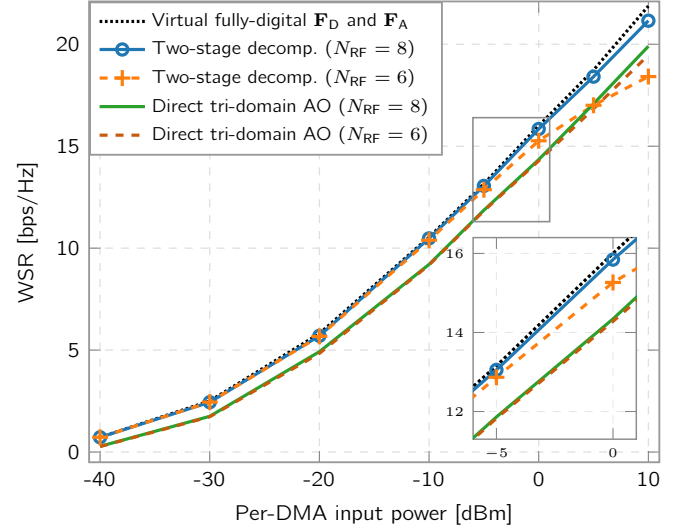

    \centering
    \include{figs/fig_twoStage.tex}
    \vspace{-2.5em}
    \caption{{WSR versus per-DMA input power for the two-stage model-based solution and the direct tri-domain alternating-optimization baseline with $K=3$ users and two streams per user. The virtual curve is obtained by solving Problem~\ref{P1} before RF-baseband decomposition, while the two-stage curves report the WSR after decomposing $\Fm_\Dt$ into $\{\Fm_\mathsf{RF},\Fm_\mathsf{BB}\}$ for different $N_\mathsf{RF}$. The direct baseline alternates over $\{\Fm_\mathsf{A},\Fm_\mathsf{RF},\Fm_\mathsf{BB}\}$.}}
    \label{fig_twoStage}
\end{figure}

To further understand the performance of the two-stage decomposition, we additionally consider a direct tri-domain alternating optimization (AO) baseline. This is a purpose-built reference implementation rather than a standard benchmark established in the literature. This baseline does not introduce the virtual precoder $\Fm_\Dt$; instead, after updating the \ac{WMMSE} auxiliary variables, it directly alternates over $\Fm_\mathsf{BB}$, $\Fm_\mathsf{RF}$, and $\Fm_\mathsf{A}$, where $\Fm_\mathsf{RF}$ and $\Fm_\mathsf{A}$ are updated through projected gradient-type steps that enforce the corresponding hardware constraints. To reduce initialization bias, all methods use three random starts and retain the best solution; the direct AO baseline also uses the previous power-point solution as a warm start. As shown in Fig.~\ref{fig_twoStage}, when $N_\mathsf{RF}=8$, the two-stage design outperforms the direct AO baseline over the entire tested power range. When $N_\mathsf{RF}=6$, the direct AO baseline exceeds the two-stage curve only in a limited high-power region, which is consistent with the amplified decomposition loss under a stringent RF-chain dimension. These results suggest that directly optimizing the three precoders, although avoiding a matrix-decomposition step, faces a more difficult nonconvex problem with strongly coupled variables and can easily converge to poor stationary points. In contrast, the two-stage formulation provides an easier optimization route that works very well in practice. Since the direct AO baseline is introduced only to examine the effect of the two-stage decomposition and does not offer superior overall performance, it is not included as a standard benchmark in the subsequent evaluations.

}

\subsection{Unfolded Model Training and Inference}

\begin{figure}[t]
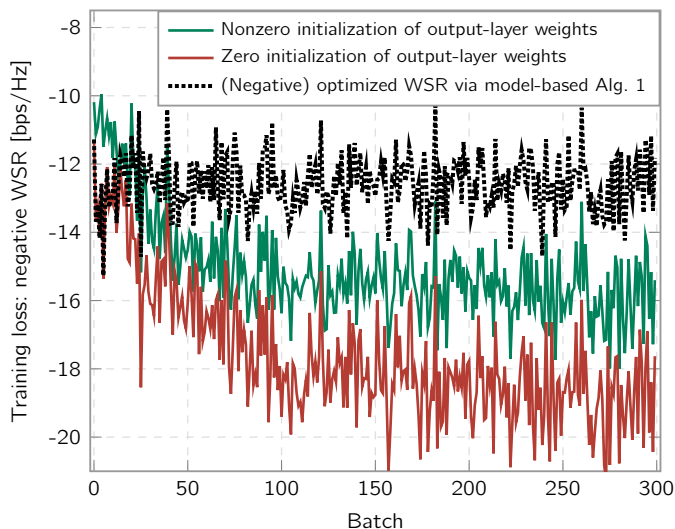

    \centering
    \include{figs/fig_BatchLoss.tex}
    \vspace{-2.5em}
    \caption{Training loss versus batch index. Each batch contains 50 problem instances, and the training loss is the negative mean \ac{WSR} over a batch.}
    \label{fig_BatchLoss}
\end{figure}

In this subsection, we assess the proposed unfolded model. We first evaluate the training performance by examining the convergence behavior of the training loss. Figure~\ref{fig_BatchLoss} compares the proposed weight initialization strategy introduced in Section~\ref{sec:MWI}, namely zero initialization of the output-layer weights, with its nonzero random initialization counterpart. All remaining network weights are initialized using Kaiming initialization~\cite{He_2015_ICCV}. As described in Section~\ref{sec:LF}, the training loss is the negative mean \ac{WSR} over each batch. Here, each batch contains 50 problem instances. We also evaluate the average optimized \ac{WSR} over the same batch using Alg.~\ref{ModelAlgo} with the same layer number $L$ and plot its negative value in Fig.~\ref{fig_BatchLoss} as a reference. Only when the training loss falls below this level does the proposed deep learning model achieve performance improvement. Both the unfolded network and Alg.~\ref{ModelAlgo} use a fixed number of iterations $L=3$. The results validate the analysis in Section~\ref{sec:MWI}, showing that zero initialization of the output-layer weights provides the unfolded algorithm with a high-quality starting point that matches the performance of the original algorithm. In contrast, random initialization yields a random starting point that is worse than the original solution, thereby making training harder. This zero-initialization design is critical for training the unfolded algorithm, as it ensures that learning improves upon a high-quality baseline instead of optimizing from a random point.

\begin{figure}[t]
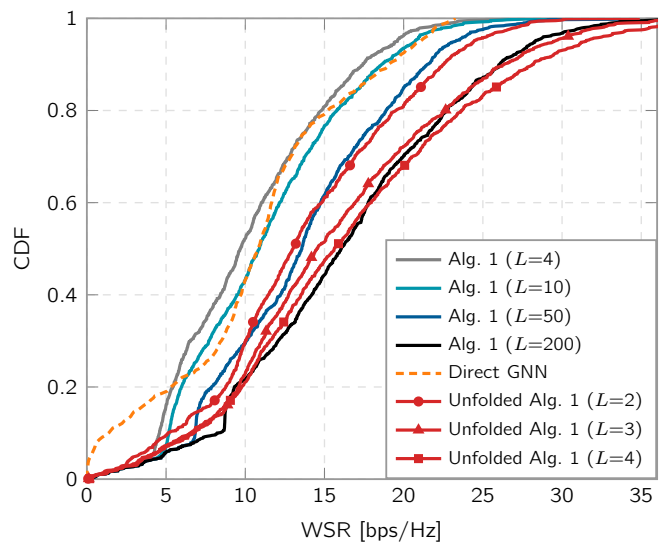

    \centering
    \include{figs/fig_CDFs.tex}
    \vspace{-2.5em}
    \caption{CDFs of the optimized \ac{WSR} on the test set with 1,000 problem instances. Alg.~\ref{ModelAlgo} with $L=\{4,10,50,200\}$ is compared with its unfolded version using $L=\{2,3,4\}$ iterations {as well as a direct GNN benchmark. The unfolded Alg.~\ref{ModelAlgo} with $L=4$ and the direct GNN benchmark contain 2,734,332 and 2,739,190 trainable parameters, respectively.}}
    \label{fig_CDFs}
\end{figure}

Next, we evaluate the acceleration performance of the proposed unfolding approach. Figure~\ref{fig_CDFs} shows the \ac{CDF} of the optimized \ac{WSR} over 1,000 problem instances in the test set. Alg.~\ref{ModelAlgo} with $L=\{4,10,50,200\}$ iterations is compared with its unfolded counterpart using $L=\{2,3,4\}$ iterations. The results indicate that, after proper training, the unfolded algorithm can substantially reduce the required number of iterations without compromising performance. For example, the unfolded algorithm with only two iterations achieves performance comparable to that of Alg.~\ref{ModelAlgo} after 50 iterations, while the four-layer unfolded algorithm attains performance similar to that of Alg.~\ref{ModelAlgo} after 200 iterations. These results demonstrate that the proposed deep unfolding architecture effectively learns to accelerate the original iterative procedure, resulting in substantial computational savings. {Importantly, this performance gain arises from the integration of model-based optimization and data-driven learning, rather than from deep learning alone. To verify this, we further implement a direct \ac{GNN} benchmark that learns the tri-domain precoders directly from channel features without exploiting the model-based algorithm. As shown in Fig.~\ref{fig_CDFs}, even though this benchmark has slightly more trainable parameters than the proposed unfolded solver with $L=4$ and is trained using the same dataset, it performs significantly worse. This highlights the advantage of the proposed model-based deep unfolding framework over a purely data-driven approach.}

\begin{figure}[t]
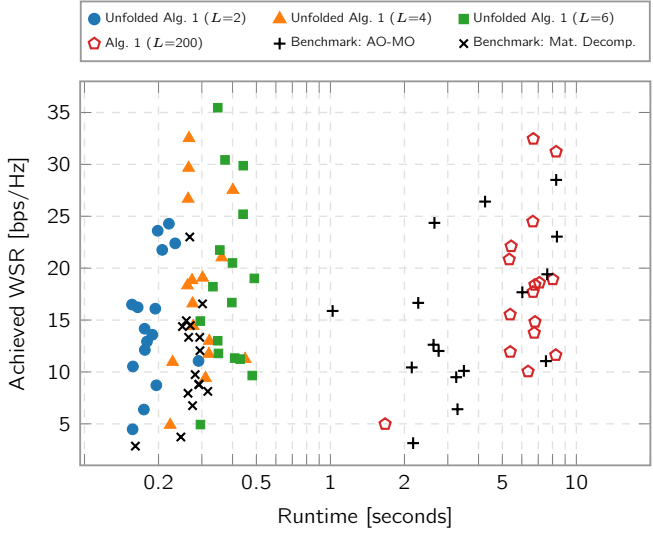

    \centering
    \include{figs/fig_WSRvsRuntime.tex}
    \vspace{-2.5em}
    \caption{Achieved WSR versus runtime over 16 sample problem instances.}
    \label{fig_WSRvsRuntime}
\end{figure} 

\begin{table}[t]
    \centering
    \setlength{\tabcolsep}{5pt}
    \caption{Achieved WSR versus runtime for different algorithms. 
    \\The values are reported as mean $\pm$ standard deviation.}
    \label{tab:runtime_wsr_comparison}
    \begin{tabular}{lcc}
        \toprule
        \textbf{Algorithm} 
        & \textbf{Runtime [s]} 
        & \textbf{WSR [bps/Hz]} \\
        \midrule
        Alg.~\ref{ModelAlgo} ($L=200$)
        & $6.342 \pm 1.522$
        & $17.960 \pm 7.057$ \\
        Unfolded Alg.~\ref{ModelAlgo} ($L=2$)
        & $0.192 \pm 0.034$
        & $14.675 \pm 5.806$ \\
        Unfolded Alg.~\ref{ModelAlgo} ($L=4$)
        & $0.300 \pm 0.059$
        & $17.868 \pm 7.681$ \\
        Unfolded Alg.~\ref{ModelAlgo} ($L=6$)
        & $0.387 \pm 0.058$
        & $18.373 \pm 8.204$  \\
        \midrule
        Benchmark: AO-MO
        & $4.231 \pm 2.387$
        & $15.450 \pm 7.122$ \\
        Benchmark: Mat. Decomp.
        & $0.271 \pm 0.034$
        & $11.174 \pm 4.947$ \\
        \bottomrule
    \end{tabular}
\end{table}

We further evaluate the runtime of the proposed algorithms to quantify their computational efficiency. We compare it with the two benchmark approaches to position our method within the literature. Figure~\ref{fig_WSRvsRuntime} plots the achieved \ac{WSR} versus the measured runtime over 16 sample problem instances, and the corresponding quantitative results are summarized in Table~\ref{tab:runtime_wsr_comparison}. The results clearly show that the unfolding design can reduce the actual runtime by more than an order of magnitude without sacrificing the achieved \ac{WSR}. In particular, increasing the number of unfolded layers leads to higher \ac{WSR} but also longer runtime, indicating a trade-off between optimization performance and execution time. Moreover, compared with the two benchmark approaches, the proposed unfolded algorithm stands out by achieving high \ac{WSR} at low runtime, which is a performance regime unattained by the two benchmarks. 

\subsection{Scalability, Robustness, and Generalization}

\begin{figure}[t]
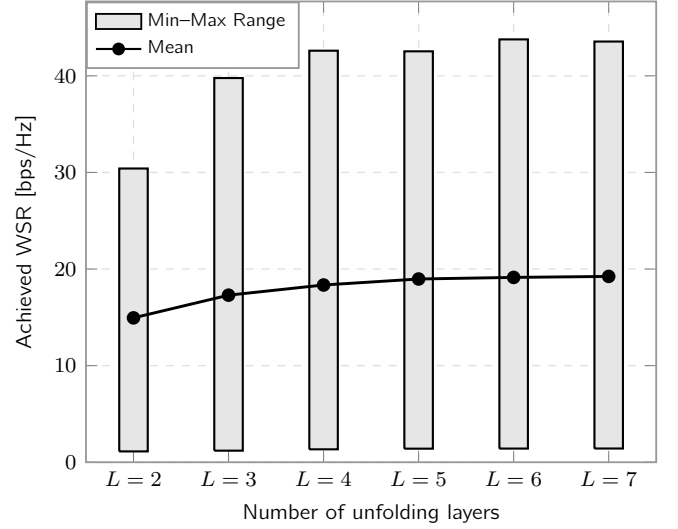

    \centering
    \include{figs/fig_WSRvsL.tex}
    \vspace{-2.5em}
    \caption{Achieved WSR versus number of unfolding layers.}
    \label{fig_WSRvsL}
\end{figure} 

Figure~\ref{fig_WSRvsL} presents the achieved \ac{WSR} versus the number of unfolding layers. Increasing the number of layers improves performance, but saturation occurs for $L>5$. Since the execution time increases linearly with the number of unfolding layers, this indicates a trade-off between performance and computational efficiency. {
We further evaluate the scalability of the proposed method \ac{w.r.t.} the array configuration, including the number of \acp{DMA} $N_\Tt$ and the number of elements per \ac{DMA} $N_\Ct$. Since the \ac{CSI} dimension in the \ac{GNN} modules depends on both $N_\Tt$ and $N_\Ct$, the unfolded networks are trained separately for each tested array configuration. The resulting WSR and runtime are summarized in Table~\ref{tab:scalability_nt}. The results show that the proposed unfolding strategy remains effective as the array configuration scales up. In general, larger array configurations lead to higher \ac{WSR} and longer runtimes for all the tested methods. 
}

\begin{table*}[t]
    \centering
    \setlength{\tabcolsep}{5pt}
    \renewcommand{\arraystretch}{1.2}
    \caption{Scalability comparison w.r.t. $N_\Tt$ and $N_\Ct$. Each entry reports the average WSR and runtime over 128 test problem instances.}
    \label{tab:scalability_nt}
    \scriptsize
    \begin{tabular}{c c c c c c c c}
        \toprule
        $N_\Tt$
        & $N_\Ct$ 
        & \makecell{Unfolded Alg.~\ref{ModelAlgo}\\($L$=2)} 
        & \makecell{Unfolded Alg.~\ref{ModelAlgo}\\($L$=4)}
        & \makecell{Unfolded Alg.~\ref{ModelAlgo}\\($L$=6)}
        & \makecell{Alg.~\ref{ModelAlgo}\\($L$=200)}
        & \makecell{Benchmark:\\AO-MO}
        & \makecell{Benchmark:\\Mat. Decomp.} \\
        \midrule
        10 & 5 & 10.01 bps/Hz, 0.049 s & 11.85 bps/Hz, 0.061 s & 12.52 bps/Hz, 0.083 s & 12.96 bps/Hz, 2.159 s & 10.57 bps/Hz, 2.382 s & 7.27 bps/Hz, 0.106 s \\
        20 & 5 & 13.35 bps/Hz, 0.178 s & 16.07 bps/Hz, 0.238 s & 16.12 bps/Hz, 0.307 s & 16.24 bps/Hz, 6.643 s & 13.44 bps/Hz, 4.339 s & 11.00 bps/Hz, 0.326 s \\
        30 & 5 & 15.26 bps/Hz, 0.292 s & 19.70 bps/Hz, 0.422 s & 20.77 bps/Hz, 0.561 s & 18.99 bps/Hz, 14.045 s & 16.42 bps/Hz, 7.320 s & 13.62 bps/Hz, 0.574 s \\
        10 & 10 & 10.45 bps/Hz, 0.039 s & 13.01 bps/Hz, 0.061 s & 14.71 bps/Hz, 0.084 s & 12.97 bps/Hz, 2.119 s & 13.43 bps/Hz, 2.969 s & 7.34 bps/Hz, 0.221 s \\
        10 & 15 & 12.91 bps/Hz, 0.041 s & 13.82 bps/Hz, 0.064 s & 17.64 bps/Hz, 0.089 s & 15.90 bps/Hz, 2.224 s & 15.93 bps/Hz, 4.141 s & 7.54 bps/Hz, 0.334 s \\
        \bottomrule
    \end{tabular}
\end{table*}

\begin{figure}[t]
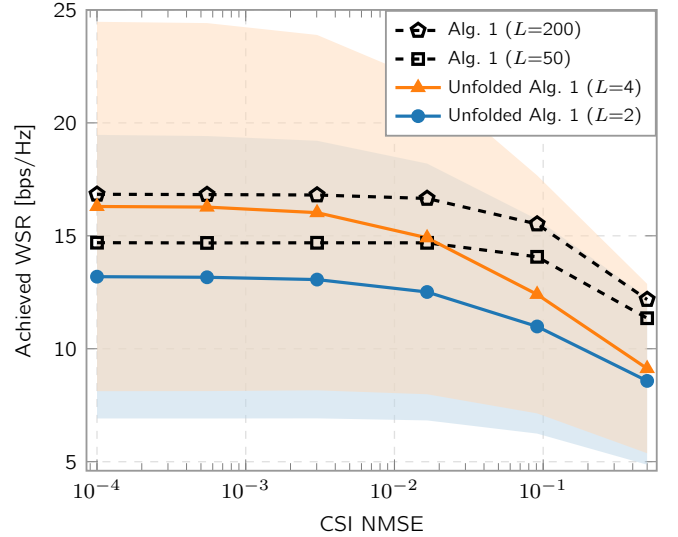

    \centering
    \include{figs/fig_WSRvsNMSE.tex}
    \vspace{-2.5em}
    \caption{Achieved WSR evaluated using noisy CSI. The curves show the mean over the test set, and the blue and orange shaded regions denote the standard deviation for the unfolded Alg.~1 with $L=2$ and $L=4$, respectively.}
    \label{fig_WSRvsNMSE}
\end{figure} 

A practical concern for the proposed method is its robustness to imperfect \ac{CSI}, as obtaining accurate \ac{CSI} remains a major challenge in real-world \ac{MIMO} systems, particularly in massive \ac{MIMO}~\cite{Ke2020Compressive}. In Fig.~\ref{fig_WSRvsNMSE}, we evaluate the achieved \ac{WSR} under noisy \ac{CSI} using Alg.~\ref{ModelAlgo} and its unfolded version. The noisy \ac{CSI} is generated as \(\hat{\Hm}_k=\Hm_k+\Em_k,\ \forall k\), where \(\Em_k\) is a full matrix whose entries are \ac{i.i.d.} complex zero-mean Gaussian random variables, with power adjusted to achieve different \ac{NMSE} levels. The \ac{CSI} \ac{NMSE} is defined as $\|\Em_k\|_\Ft^2/\|\Hm_k\|_\Ft^2$. The results show that, for both the original iterative Alg.~\ref{ModelAlgo} and its unfolded version, increasing the number of iterations reduces robustness to \ac{CSI} errors. This effect is more pronounced in the unfolded algorithm. For example, the unfolded algorithm with $L=2$ exhibits robustness comparable to that of the original Alg.~\ref{ModelAlgo} with $L=200$, whereas the robustness of the unfolded algorithm with $L=4$ degrades substantially. These findings indicate that the proposed deep unfolding method is more effective when accurate \ac{CSI} is available. Nevertheless, the performance degradation due to inaccurate \ac{CSI} remains limited, while the advantage of higher computational efficiency is consistently maintained.

\begin{figure}[t]
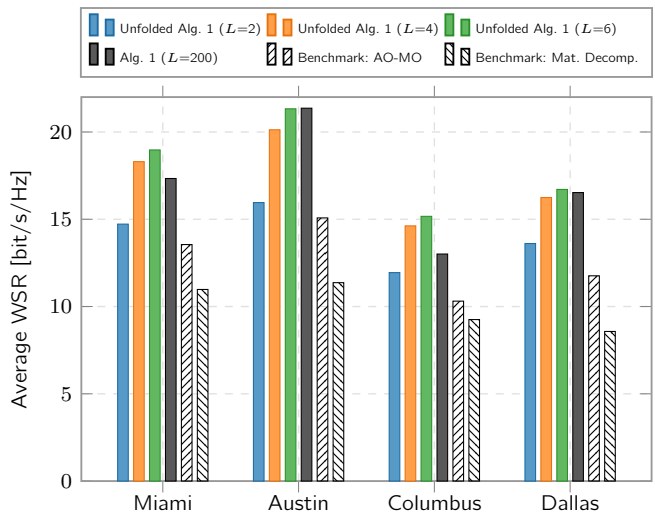

    \centering
    \include{figs/fig_Generalization.tex}
    \vspace{-2.5em}
    \caption{Achieved WSR evaluated using channel datasets from different cities generated by the DeepMIMO platform.}
    \label{fig_Generalization}
\end{figure} 

Finally, we assess the generalization capability of the proposed method. In Fig.~\ref{fig_Generalization}, we compare the proposed Alg.~\ref{ModelAlgo}, its unfolded version, and the two benchmark schemes on channel datasets generated from Miami (the predefined test set) and three additional cities, namely Austin, Columbus, and Dallas, within the DeepMIMO platform. All these datasets are different from those used for training, as described in Section~\ref{sec:TS}.  The results clearly show that both the proposed Alg.~\ref{ModelAlgo} and its unfolded version consistently achieve strong performance and outperform the benchmarks across all city scenarios, demonstrating robust generalization capability. Although Alg.~\ref{ModelAlgo} also achieves high WSR, the significantly reduced runtime makes the unfolded Alg.~\ref{ModelAlgo} a more practical solution for real-time deployment.

\section{Limitations}

Although the effectiveness of the proposed deep unfolding method has been demonstrated for tri-hybrid beamforming, several limitations should be acknowledged. First, the method requires offline training before deployment. Furthermore, for tri-hybrid beamforming systems employing different antenna technologies or system architectures, the unfolding network must be redesigned and retrained accordingly. No unified solution currently exists for all tri-hybrid beamforming scenarios.

{In addition, while this paper focuses on a specific \ac{DMA}-based tri-hybrid architecture, the proposed algorithm may experience performance degradation in practical implementations due to the omission of certain realistic hardware considerations. For example, practical \acp{DMA} exhibit frequency-selective responses, making wideband effects important, particularly in high-frequency communications~\cite{Carlson2026Wideband,Deshpande2026Frequency}. Moreover, \ac{DMA} reconfiguration is typically achieved by tuning the resonant characteristics (e.g., the resonant frequency or effective impedance) of individual elements using components such as varactors or PIN diodes~\cite{PhysRevApplied}. Consequently, the achievable element responses are often confined to a limited feasible region or finite-resolution tuning, rather than allowing arbitrary phase-shift control. Incorporating these practical hardware constraints into the algorithm design is an important direction for future research.
}

\section{Conclusion}

While tri-hybrid MIMO architectures offer significant potential for flexible and efficient communications, their multi-domain structure poses substantial challenges for beamforming optimization. The strong coupling among design variables and hardware constraints makes the problem inherently complex and computationally intensive. By studying a representative multiuser tri-hybrid system employing dynamic metasurface antennas, this paper shows that deep algorithm unfolding provides an effective means of reducing this computational burden, thereby facilitating real-time implementation. The effectiveness of the proposed approach relies on carefully designed trainable networks that integrate model-based optimization with data-driven learning. {Moreover, this study reveals a conceptual shift in algorithm design under the deep unfolding framework, where the asymptotic convergence behavior of the original iterative solver becomes less critical, while low per-iteration complexity and an iteration structure amenable to unfolding emerge as more important design considerations for developing a fast solver.}

{\appendix
\renewcommand{\theequation}{A.\arabic{equation}}
\setcounter{equation}{0}
\section{Proof of \eqref{eq:un_upd}}
{
In this proof, we omit the \ac{DMA} index $n$ for brevity. Let $\uv^{(t)}$ denote a feasible point of $\uv_n$ in~\eqref{eq:optun}, and write the objective function as
\begin{equation}
	f(\uv)=\uv^\HH \Mm \uv + 2\mathrm{Re}\{(\bv/2)^\HH \uv\}.
\end{equation}
This is a unit-modulus quadratic program in the form considered in~\cite{He2022QCQP}. Applying the quadratic majorization in~\cite[Eqs.~(2)--(3)]{He2022QCQP} with a $\rho\geq\lambda_\mathsf{max}(\Mm)$ gives the following linear surrogate over $|[\uv]_m|=1$:
\begin{multline}
	g\big(\uv|\uv^{(t)}\big)
	=
	2\mathrm{Re}\Big\{\uv^\HH\Big((\Mm-\rho\mathbf{I})\uv^{(t)}+\frac{1}{2}\bv\Big)\Big\}
	\\ +2\rho N_\Ct-\big(\uv^{(t)}\big)^\HH\Mm\uv^{(t)}.
\end{multline}
The constant-modulus minimizer of this surrogate follows from the phase-projection solution in~\cite[eq.~(6)]{He2022QCQP}, namely
\begin{equation}
	\uv^{(t+1)}
	=
	e^{\mathsf{j}\arg\big((\rho \mathbf{I} - \Mm)\uv^{(t)} - \frac{1}{2}\bv\big)}.
\end{equation}
Therefore, \eqref{eq:un_upd} is one majorization-minimization update for~\eqref{eq:optun}. Since $g(\uv|\uv^{(t)})$ upper-bounds $f(\uv)$ and is tight at $\uv^{(t)}$, the standard majorization-minimization descent property~\cite[Eqs.~(2)--(5)]{Breloy2021Majorization} guarantees
\begin{equation}
	f\big(\uv^{(t+1)}\big) \leq g\big(\uv^{(t+1)}|\uv^{(t)}\big) \leq g\big(\uv^{(t)}|\uv^{(t)}\big) = f\big(\uv^{(t)}\big).\notag
\end{equation}
This completes the proof.
}
}

\section{Acknowledgment}
The authors would like to thank Prof.~Robert~W.~Heath~Jr.\ for his helpful comments on the system modeling.

\bibliography{references}
\bibliographystyle{IEEEtran}

\end{document}

%% file: preamble.tex
\usepackage{amsmath,amssymb,amsfonts}
\usepackage{array}
\usepackage[caption=false,font=footnotesize,labelfont=sf,textfont=sf]{subfig}
\usepackage{textcomp}
\usepackage{stfloats}
\usepackage{verbatim}
\usepackage{graphicx}
\usepackage{cite}
\usepackage{tabularx}
\usepackage{booktabs}

\def\BibTeX{{\rm B\kern-.05em{\sc i\kern-.025em b}\kern-.08em
    T\kern-.1667em\lower.7ex\hbox{E}\kern-.125emX}}

\usepackage{acronym}
\usepackage{acronym}
\acrodef{5G}{the fifth generation}
\acrodef{MIMO}{multiple-input multiple-output}
\acrodef{SISO}{single-input single-output}
\acrodef{MISO}{multiple-input single-output}
\acrodef{MU}{multiuser}
\acrodef{RF}{radio frequency}
\acrodef{LoS}{line-of-sight}
\acrodef{NLoS}{non-line-of-sight}
\acrodef{AoA}{angle-of-arrival}
\acrodef{AoD}{angle-of-departure}
\acrodef{UPA}{uniform planar array}
\acrodef{ARV}{array response vector}
\acrodef{EM}{electromagnetic}
\acrodef{MA}{movable antenna}
\acrodef{BS}{base station}
\acrodef{UE}{user equipment}
\acrodef{ERA}{electromagnetically reconfigurable antenna}
\acrodef{AWGN}{additive white Gaussian noise}
\acrodef{HFSS}{High-Frequency Structure Simulator}
\acrodef{BCD}{block coordinate descent}
\acrodef{ZF}{zero-forcing}
\acrodef{SNR}{Signal-to-Noise Ratio}
\acrodef{6G}{sixth generation}
\acrodef{3D}{three-dimensional}
\acrodef{2D}{two-dimensional}
\acrodef{RIS}{reconfigurable intelligent surface}
\acrodef{DMA}{dynamic metasurface antenna}
\acrodef{PA}{power amplifier}
\acrodef{WMMSE}{weighted sum-minimum mean-square error}
\acrodef{WSR}{weighted sum-rate}
\acrodef{GNN}{graph neural network}
\acrodef{GraphSAGE}{graph sample and aggregate}
\acrodef{w.r.t.}{with respect to}
\acrodef{MLP}{multi-layer perceptron}
\acrodef{CDF}{cumulative distribution function}
\acrodef{CSI}{channel state information}
\acrodef{i.i.d.}{independent and identically distributed}
\acrodef{NMSE}{normalized mean squared error}
\newtheorem{lemma}{\textbf{Lemma}}

\newenvironment{proof}{\textit{\textbf{Proof:}}}{\hfill$\square$}
\usepackage{color} 
\usepackage[normalem]{ulem}

\usepackage{tikz} 
\usepackage[utf8]{inputenc}
\usepackage{pgfplots}
\usepackage{tabu,longtable}
\usepackage{diagbox}
\usepackage{threeparttable}
\usepackage{makecell}
\usepackage{multirow} 
\usepackage{units} 		
\usepackage{bm} 		
\usepackage{amssymb} 	
\input{symbol}

\usepackage{hyperref} 	
\usepackage{mathrsfs}   
\usepackage{amsmath}

\usepackage{algorithm}
\usepackage{algpseudocode}

\makeatletter
\algnewcommand{\LineComment}[1]{\Statex \hskip\ALG@thistlm \(\triangleright\) #1}
\makeatother

\usepackage{pgfplots}       
\pgfplotsset{compat=1.18}   
\usepackage{xcolor}         
\usepgfplotslibrary{fillbetween}
\usetikzlibrary{patterns}
\definecolor{profBlue}{RGB}{0,91,150}        
\definecolor{profCyan}{RGB}{0,150,170}       
\definecolor{profPurple}{RGB}{120,60,180}    
\definecolor{profOrange}{RGB}{200,90,10}     
\definecolor{profGray}{RGB}{90,90,90}        
\definecolor{profGreen}{RGB}{0,130,90}       
\definecolor{profRed}{RGB}{180,60,50}        

\definecolor{tabblue}{RGB}{31,119,180}
\definecolor{taborange}{RGB}{255,127,14}
\definecolor{tabgreen}{RGB}{44,160,44}
\definecolor{tabred}{RGB}{214,39,40}

%% file: symbol.tex
\newcommand{\TT}{\mathsf{T}}
\newcommand{\HH}{\mathsf{H}}

\newcommand{\av}{{\bf a}}
\newcommand{\bv}{{\bf b}}

\newcommand{\dv}{{\bf d}}

\newcommand{\nv}{{\bf n}}

\newcommand{\sv}{{\bf s}}

\newcommand{\uv}{{\bf u}}
\newcommand{\wv}{{\bf w}}
\newcommand{\vv}{{\bf v}}

\newcommand{\yv}{{\bf y}}
\newcommand{\zv}{{\bf z}}

\newcommand{\Am}{{\bf A}}
\newcommand{\Bm}{{\bf B}}

\newcommand{\Dm}{{\bf D}}
\newcommand{\Em}{{\bf E}}
\newcommand{\Fm}{{\bf F}}
\newcommand{\Gm}{{\bf G}}
\newcommand{\Hm}{{\bf H}}

\newcommand{\Mm}{{\bf M}}
\newcommand{\Nm}{{\bf N}}

\newcommand{\Qm}{{\bf Q}}

\newcommand{\Wm}{{\bf W}}

\newcommand{\Xm}{{\bf X}}
\newcommand{\Ym}{{\bf Y}}

\newcommand{\At}{{\mathsf A}}

\newcommand{\Ct}{{\mathsf C}}
\newcommand{\Dt}{{\mathsf D}}

\newcommand{\Ft}{{\mathsf F}}

\newcommand{\St}{{\mathsf S}}
\newcommand{\Tt}{{\mathsf T}}

\newcommand{\betav}{\ensuremath{\bm{\beta}}}

\newcommand{\etav}{\ensuremath{\bm{\eta}}}

\newcommand{\thetav}{\ensuremath{\bm{\theta}}}

\newcommand{\Gammam}{\ensuremath{\bm{\Gamma}}}

\newcommand{\Thetam}{\ensuremath{\bm{\Theta}}}
\newcommand{\Omegam}{\ensuremath{\bm{\Omega}}}

%% file: figs/fig_P1.tex
\renewcommand{\sfdefault}{cmbr} 
\newcommand{\ttk}[1]{\footnotesize\textsf{#1}} 

\pgfplotsset{
    every axis/.append style={
        label style={font=\sffamily\footnotesize},      
        legend style={font=\sffamily\footnotesize},     
        tick label style={font=\sffamily\footnotesize}  
    }
}

\begin{tikzpicture}
\begin{axis}[%
    width=2.65in, height=1.45in, scale only axis,
    xmin=-1, xmax=201, ymin=0, ymax=29,
    xlabel={Iteration $L$}, 
    ylabel={Weighted Sum-Rate [bps/Hz]},
    xtick={0,20,40,60,80,100,120,140,160,180,200},
    ytick={0,5,10,15,20,25},
    xticklabels={\ttk{0},\ttk{20},\ttk{40},\ttk{60},\ttk{80},\ttk{100},\ttk{120},\ttk{140},\ttk{160},\ttk{180},\ttk{200}},
    yticklabels={\ttk{0},\ttk{5},\ttk{10},\ttk{15},\ttk{20},\ttk{25}},
    axis line style={line width=0.7pt, draw=black!40},
    axis background/.style={fill=white},
    xmajorgrids=true,
    ymajorgrids=true,
    grid style={dashed, line width=0.5pt, draw=black!12},
    legend style={
        at={(1,1)},
        anchor=north east,
        legend cell align=left,
        align=left,
        fill=white,
        draw=black!40,
        line width=0.7pt,
        font=\sffamily\scriptsize}
]

\addplot[
    line width=1.2pt,
    color=profBlue
]table[row sep=crcr]{%
0	0.499788\\
1	7.768129\\
2	8.152795\\
3	8.485914\\
4	8.783383\\
5	9.052294\\
6	9.297714\\
7	9.523500\\
8	9.732557\\
9	9.927252\\
10	10.109440\\
11	10.280649\\
12	10.442153\\
13	10.594989\\
14	10.740068\\
15	10.878162\\
16	11.009984\\
17	11.136257\\
18	11.257990\\
19	11.377235\\
20	11.499325\\
21	11.639288\\
22	11.837286\\
23	12.176438\\
24	12.731006\\
25	13.412243\\
26	14.034561\\
27	14.528696\\
28	14.923336\\
29	15.258572\\
30	15.552918\\
31	15.815597\\
32	16.055103\\
33	16.275856\\
34	16.481035\\
35	16.673088\\
36	16.853874\\
37	17.024952\\
38	17.187708\\
39	17.342760\\
40	17.490919\\
41	17.632942\\
42	17.769369\\
43	17.900570\\
44	18.026974\\
45	18.149073\\
46	18.267292\\
47	18.381857\\
48	18.493095\\
49	18.601137\\
50	18.706276\\
51	18.808664\\
52	18.908445\\
53	19.005793\\
54	19.100183\\
55	19.192291\\
56	19.282372\\
57	19.370405\\
58	19.456608\\
59	19.540993\\
60	19.623652\\
61	19.704681\\
62	19.784140\\
63	19.862061\\
64	19.938566\\
65	20.013615\\
66	20.087341\\
67	20.159767\\
68	20.230972\\
69	20.300934\\
70	20.369713\\
71	20.437363\\
72	20.503918\\
73	20.569447\\
74	20.633892\\
75	20.697365\\
76	20.759907\\
77	20.821442\\
78	20.882065\\
79	20.941814\\
80	21.000671\\
81	21.058699\\
82	21.115919\\
83	21.172293\\
84	21.227922\\
85	21.282688\\
86	21.336794\\
87	21.390144\\
88	21.442751\\
89	21.494652\\
90	21.545858\\
91	21.596401\\
92	21.646244\\
93	21.695465\\
94	21.744068\\
95	21.792011\\
96	21.839346\\
97	21.886063\\
98	21.932066\\
99	21.977177\\
100	22.021795\\
101	22.065788\\
102	22.109226\\
103	22.152079\\
104	22.194370\\
105	22.236191\\
106	22.277397\\
107	22.318130\\
108	22.358248\\
109	22.397947\\
110	22.437109\\
111	22.475805\\
112	22.514009\\
113	22.551737\\
114	22.589006\\
115	22.625843\\
116	22.662233\\
117	22.698128\\
118	22.733662\\
119	22.768726\\
120	22.803455\\
121	22.837706\\
122	22.871548\\
123	22.905020\\
124	22.938114\\
125	22.970907\\
126	23.003254\\
127	23.035217\\
128	23.066935\\
129	23.098213\\
130	23.129135\\
131	23.159760\\
132	23.190056\\
133	23.220009\\
134	23.249638\\
135	23.278999\\
136	23.308002\\
137	23.336626\\
138	23.365078\\
139	23.393154\\
140	23.420956\\
141	23.448404\\
142	23.475643\\
143	23.502491\\
144	23.529118\\
145	23.555458\\
146	23.581461\\
147	23.607235\\
148	23.632694\\
149	23.657856\\
150	23.682831\\
151	23.707413\\
152	23.731857\\
153	23.755936\\
154	23.779713\\
155	23.803356\\
156	23.826603\\
157	23.849636\\
158	23.872419\\
159	23.894884\\
160	23.917200\\
161	23.939220\\
162	23.960964\\
163	23.982431\\
164	24.003693\\
165	24.024752\\
166	24.045410\\
167	24.066053\\
168	24.086224\\
169	24.106356\\
170	24.126146\\
171	24.145763\\
172	24.165127\\
173	24.184301\\
174	24.203243\\
175	24.221973\\
176	24.240379\\
177	24.258678\\
178	24.276768\\
179	24.294634\\
180	24.312374\\
181	24.329857\\
182	24.347183\\
183	24.364246\\
184	24.381176\\
185	24.397955\\
186	24.414511\\
187	24.430817\\
188	24.447016\\
189	24.463047\\
190	24.478966\\
191	24.494640\\
192	24.510197\\
193	24.525537\\
194	24.540842\\
195	24.555826\\
196	24.570744\\
197	24.585533\\
198	24.600128\\
199	24.614763\\
200	24.628988\\
};

\def\ymaxval{24.628988}

\addplot[
    dashed,
    line width=1pt,
    color=black!60,
    domain=-1:201
] {\ymaxval};

\node[anchor=south, yshift=0pt, font=\sffamily\scriptsize]
at (axis cs:42, 24.5) {Max = 24.63 bps/Hz};

\draw[->, thick, color=black!60, line width=1pt] 
  (axis cs:40,17.490919-0.5) -- ++(axis direction cs:0,-7);
\node[anchor=north, font=\sffamily\scriptsize] 
  at (axis cs:40,17.490919-8.2) {1.07 s};

\draw[->, thick, color=black!60, line width=1pt] 
  (axis cs:100,22.021795-0.5) -- ++(axis direction cs:0,-7);
\node[anchor=north, font=\sffamily\scriptsize] 
  at (axis cs:100,22.021795-8.2) {2.62 s};

\draw[->, thick, color=black!60, line width=1pt] 
  (axis cs:160,23.917200-0.5) -- ++(axis direction cs:0,-7);
\node[anchor=north, font=\sffamily\scriptsize] 
  at (axis cs:160,23.917200-8.2) {4.17 s};

\end{axis}
\end{tikzpicture}

%% file: figs/fig_P2.tex
\renewcommand{\sfdefault}{cmbr} 
\newcommand{\ttk}[1]{\footnotesize\textsf{#1}} 

\pgfplotsset{
    every axis/.append style={
        label style={font=\sffamily\footnotesize},      
        legend style={font=\sffamily\footnotesize},     
        tick label style={font=\sffamily\footnotesize}  
    }
}

\begin{tikzpicture}
\begin{axis}[%
    width=2.65in, height=1.45in, scale only axis,
    xmin=-0.5, xmax=60, ymin=22.5, ymax=25.2,
    xlabel={Iteration $L$}, 
    ylabel={Weighted Sum-Rate [bps/Hz]},
    xtick={0,10,20,30,40,50,60},
    ytick={22.5,23,23.5,24,24.5,25},
    xticklabels={\ttk{0},\ttk{10},\ttk{20},\ttk{30},\ttk{40},\ttk{50},\ttk{60}},
    yticklabels={\ttk{22.5},\ttk{23},\ttk{23.5},\ttk{24},\ttk{24.5},\ttk{25}},
    axis line style={line width=0.7pt, draw=black!40},
    axis background/.style={fill=white},
    xmajorgrids=true,
    ymajorgrids=true,
    grid style={dashed, line width=0.5pt, draw=black!12},
    legend style={
        at={(1,1)},
        anchor=north east,
        legend cell align=left,
        align=left,
        fill=white,
        draw=black!40,
        line width=0.7pt,
        font=\sffamily\scriptsize}
]

\addplot[
    line width=1.2pt,
    color=profBlue
]table[row sep=crcr]{%
0	22.614822\\
1	23.036993\\
2	23.337839\\
3	23.644596\\
4	23.812769\\
5	24.139257\\
6	24.178812\\
7	24.230724\\
8	24.158266\\
9	24.336872\\
10	24.431890\\
11	24.403877\\
12	24.440231\\
13	24.455576\\
14	24.444054\\
15	24.471277\\
16	24.483503\\
17	24.480045\\
18	24.502848\\
19	24.508657\\
20	24.514721\\
21	24.520741\\
22	24.526936\\
23	24.536337\\
24	24.552116\\
25	24.559155\\
26	24.571808\\
27	24.573280\\
28	24.578560\\
29	24.586367\\
30	24.582876\\
31	24.591007\\
32	24.588507\\
33	24.597597\\
34	24.596672\\
35	24.601065\\
36	24.601921\\
37	24.605297\\
38	24.605051\\
39	24.607872\\
40	24.610920\\
41	24.612116\\
42	24.610870\\
43	24.615242\\
44	24.611923\\
45	24.615265\\
46	24.615196\\
47	24.615036\\
48	24.616579\\
49	24.617907\\
50	24.618809\\
51	24.619644\\
52	24.620430\\
53	24.621368\\
54	24.622177\\
55	24.622513\\
56	24.622936\\
57	24.623547\\
58	24.624054\\
59	24.624306\\
};

\def\ymaxval{24.624306}

\addplot[
    dashed,
    line width=1.2pt,
    color=black!60,
    domain=-1:201
] {\ymaxval};

\node[anchor=south, yshift=1.5pt, font=\sffamily\scriptsize]
at (axis cs:12.5, 24.624) {Max = 24.62 bps/Hz};

\draw[->, thick, color=black!60, line width=1pt] 
  (axis cs:15,24.471277-0.05) -- ++(axis direction cs:0,-0.7);
\node[anchor=north, font=\sffamily\scriptsize] 
  at (axis cs:15,24.471277-0.8) {0.033 s};

\draw[->, thick, color=black!60, line width=1pt] 
  (axis cs:30,24.582876-0.05) -- ++(axis direction cs:0,-0.7);
\node[anchor=north, font=\sffamily\scriptsize] 
  at (axis cs:30,24.582876-0.8) {0.047 s};

\draw[->, thick, color=black!60, line width=1pt] 
  (axis cs:45,24.615265-0.05) -- ++(axis direction cs:0,-0.7);
\node[anchor=north, font=\sffamily\scriptsize] 
  at (axis cs:45,24.615265-0.8) {0.062 s};

\end{axis}
\end{tikzpicture}

%% file: figs/fig_twoStage.tex
\renewcommand{\sfdefault}{cmbr} 
\newcommand{\ttk}[1]{\footnotesize\textsf{#1}} 

\pgfplotsset{
    every axis/.append style={
        label style={font=\sffamily\footnotesize},      
        legend style={font=\sffamily\footnotesize},     
        tick label style={font=\sffamily\footnotesize}  
    }
}

\begin{tikzpicture}
\begin{axis}[%
    name=mainplot,
    width=2.97in, height=2.4in, scale only axis,
    xmin=-41, xmax=11, ymin=-0.5, ymax=22.1,
    xlabel={Per-DMA input power [dBm]}, 
    ylabel={WSR [bps/Hz]},
    xtick={-40,-30,-20,-10,0,10},
    ytick={0,5,10,15,20},
    xticklabels={\ttk{-40},\ttk{-30},\ttk{-20},\ttk{-10},\ttk{0},\ttk{10}},
    yticklabels={\ttk{0},\ttk{5},\ttk{10},\ttk{15},\ttk{20}}, 
    axis line style={line width=0.7pt, draw=black!40},
    axis background/.style={fill=white},
    xmajorgrids=true,
    ymajorgrids=true,
    grid style={dashed, line width=0.5pt, draw=black!12},
    legend style={
        at={(0,1)},
        anchor=north west,
        legend cell align=left,
        align=left,
        fill=white,
        draw=black!40,
        line width=0.7pt,
        font=\sffamily\scriptsize}
]
\addplot[
    line width=1pt,
    color=black,
    densely dotted,
    line width=1.2pt,
]table[row sep=crcr]{%
-40	0.7242908654\\
-30	2.5030686498\\
-20	5.8116083740\\
-10	10.5407443603\\
-5	13.1551568639\\
0	16.0024814112\\
5	18.7345642578\\
10	21.8756926209\\
};
\addlegendentry{Virtual fully-digital $\Fm_\Dt$ and $\Fm_\At$}

\addplot[
    line width=1pt,
    color=tabblue,
    mark=o,
    mark size=2.1pt,
    line width=1.2pt,
    mark options={solid, fill=tabblue}
]table[row sep=crcr]{%
-40	0.7242099255\\
-30	2.4443287493\\
-20	5.6997768825\\
-10	10.4790131126\\
-5	13.0529077977\\
0	15.8437830741\\
5	18.4102584738\\
10	21.1549917758\\
};
\addlegendentry{Two-stage decomp. ($N_{\mathsf{RF}}=8$)}

\addplot[
    line width=1pt,
    color=taborange,
    dashed,
    mark=+,
    mark size=3pt,
    line width=1.2pt,
    mark options={solid, fill=taborange}
]table[row sep=crcr]{%
-40	0.7242277687\\
-30	2.4494505726\\
-20	5.6726318716\\
-10	10.3879585492\\
-5	12.8658557883\\
0	15.2654011473\\
5	17.0068022329\\
10	18.4157859944\\
};
\addlegendentry{Two-stage decomp. ($N_{\mathsf{RF}}=6$)}

\addplot[
    line width=1pt,
    color=tabgreen,
    line width=1.2pt,
    mark options={solid, fill=tabblue}
]table[row sep=crcr]{%
-40	0.2807776051\\
-30	1.7522492000\\
-20	4.9221003000\\
-10	9.2041143982\\
-5	11.8624762708\\
0	14.3427394455\\
5	17.0807131055\\
10	19.9020130630\\
};
\addlegendentry{Direct tri-domain AO ($N_{\mathsf{RF}}=8$)}

\addplot[
    line width=1pt,
    color=profOrange,
    dashed,
    line width=1.2pt,
    mark options={solid, fill=taborange}
]table[row sep=crcr]{%
-40	0.2686090289\\
-30	1.7298707911\\
-20	4.8366304660\\
-10	9.1718713355\\
-5	11.8276228729\\
0	14.2804024326\\
5	16.8363902103\\
10	19.4691567120\\
};
\addlegendentry{Direct tri-domain AO ($N_{\mathsf{RF}}=6$)}

\draw[black!45, line width=0.6pt]
    (axis cs:-6,11.3) rectangle (axis cs:1,16.4);
    
\end{axis}

\begin{axis}[%
    at={(mainplot.south east)},
    anchor=south east,
    xshift=-0.12in,
    yshift=0.12in,
    width=0.85in, height=1.05in, scale only axis,
    xmin=-6, xmax=1, ymin=11.3, ymax=16.4,
    xtick={-5,0},
    ytick={12,14,16},
    tick label style={font=\sffamily\tiny},
    axis line style={line width=0.55pt, draw=black!45},
    axis background/.style={fill=white},
    xmajorgrids=true,
    ymajorgrids=true,
    grid style={dashed, line width=0.35pt, draw=black!10},
]
\addplot[
    color=black,
    densely dotted,
    line width=1.2pt,
]table[row sep=crcr]{%
-40	0.7242908654\\
-30	2.5030686498\\
-20	5.8116083740\\
-10	10.5407443603\\
-5	13.1551568639\\
0	16.0024814112\\
5	18.7345642578\\
10	21.8756926209\\
};

\addplot[
    color=tabblue,
    mark=o,
    mark size=2.1pt,
    line width=1.2pt,
    mark options={solid, fill=tabblue}
]table[row sep=crcr]{%
-40	0.7242099255\\
-30	2.4443287493\\
-20	5.6997768825\\
-10	10.4790131126\\
-5	13.0529077977\\
0	15.8437830741\\
5	18.4102584738\\
10	21.1549917758\\
};

\addplot[
    color=tabgreen,
    line width=1.2pt,
    mark options={solid, fill=tabblue}
]table[row sep=crcr]{%
-40	0.2807776051\\
-30	1.7522492000\\
-20	4.9221003000\\
-10	9.2041143982\\
-5	11.8624762708\\
0	14.3427394455\\
5	17.0807131055\\
10	19.9020130630\\
};

\addplot[
    color=taborange,
    dashed,
    mark=+,
    mark size=3pt,
    line width=1.2pt,
    mark options={solid, fill=taborange}
]table[row sep=crcr]{%
-40	0.7242277687\\
-30	2.4494505726\\
-20	5.6726318716\\
-10	10.3879585492\\
-5	12.8658557883\\
0	15.2654011473\\
5	17.0068022329\\
10	18.4157859944\\
};

\addplot[
    color=profOrange,
    dashed,
    line width=1.2pt,
    mark options={solid, fill=taborange}
]table[row sep=crcr]{%
-40	0.2686090289\\
-30	1.7298707911\\
-20	4.8366304660\\
-10	9.1718713355\\
-5	11.8276228729\\
0	14.2804024326\\
5	16.8363902103\\
10	19.4691567120\\
};
\end{axis}
\end{tikzpicture}

%% file: figs/fig_BatchLoss.tex
\renewcommand{\sfdefault}{cmbr} 
\newcommand{\ttk}[1]{\footnotesize\textsf{#1}} 

\pgfplotsset{
    every axis/.append style={
        label style={font=\sffamily\footnotesize},      
        legend style={font=\sffamily\footnotesize},     
        tick label style={font=\sffamily\footnotesize}  
    }
}

\begin{tikzpicture}
\begin{axis}[%
    width=2.97in, height=2.4in, scale only axis,
    xmin=-2, xmax=302, ymin=-21, ymax=-7.5,
    xlabel={Batch}, 
    ylabel={Training loss: negative WSR [bps/Hz]},
    xtick={0,50,100,150,200,250,300},
    ytick={-22,-20,-18,-16,-14,-12,-10,-8,-6},
    xticklabels={\ttk{0},\ttk{50},\ttk{100},\ttk{150},\ttk{200},\ttk{250},\ttk{300}},
    yticklabels={\ttk{-22},\ttk{-20},\ttk{-18},\ttk{-16},\ttk{-14},\ttk{-12},\ttk{-10},\ttk{-8},\ttk{-6}}, 
    axis line style={line width=0.7pt, draw=black!40},
    axis background/.style={fill=white},
    xmajorgrids=true,
    ymajorgrids=true,
    grid style={dashed, line width=0.5pt, draw=black!12},
    legend style={
        at={(1,1)},
        anchor=north east,
        legend cell align=left,
        align=left,
        fill=white,
        draw=black!40,
        line width=0.7pt,
        font=\sffamily\scriptsize}
]
\addplot[
    line width=1pt,
    color=profGreen,
]table[row sep=crcr]{%
0	-10.187017440795898\\
1	-11.122885704040527\\
2	-10.918157577514648\\
3	-10.383801460266113\\
4	-9.951119422912598\\
5	-11.493356704711914\\
6	-11.033282279968262\\
7	-10.660650253295898\\
8	-11.172776222229004\\
9	-10.768385887145996\\
10	-11.817366600036621\\
11	-10.753162384033203\\
12	-11.592574119567871\\
13	-12.440526008605957\\
14	-12.01171588897705\\
15	-11.401972770690918\\
16	-11.986245155334473\\
17	-11.394588470458984\\
18	-12.579156875610352\\
19	-12.936415672302246\\
20	-10.22474193572998\\
21	-11.926039695739746\\
22	-12.035187721252441\\
23	-13.66046142578125\\
24	-11.333725929260254\\
25	-15.1720552444458\\
26	-12.924620628356934\\
27	-12.075475692749023\\
28	-12.98600959777832\\
29	-14.3414306640625\\
30	-13.333361625671387\\
31	-14.255300521850586\\
32	-14.483490943908691\\
33	-13.346718788146973\\
34	-12.572265625\\
35	-14.783565521240234\\
36	-13.719901084899902\\
37	-13.513313293457031\\
38	-14.606616973876953\\
39	-11.481501579284668\\
40	-12.685674667358398\\
41	-14.25533676147461\\
42	-15.411110877990723\\
43	-14.719228744506836\\
44	-13.744505882263184\\
45	-15.541858673095703\\
46	-15.015918731689453\\
47	-14.216866493225098\\
48	-14.478506088256836\\
49	-14.740152359008789\\
50	-14.50958251953125\\
51	-13.985088348388672\\
52	-14.127350807189941\\
53	-15.732382774353027\\
54	-14.03634262084961\\
55	-14.376703262329102\\
56	-15.143407821655273\\
57	-14.146860122680664\\
58	-15.739723205566406\\
59	-15.090963363647461\\
60	-15.023324966430664\\
61	-14.798613548278809\\
62	-15.451106071472168\\
63	-13.687479019165039\\
64	-15.351470947265625\\
65	-14.13085651397705\\
66	-14.540709495544434\\
67	-16.01726531982422\\
68	-13.868908882141113\\
69	-15.911327362060547\\
70	-13.314746856689453\\
71	-14.270726203918457\\
72	-15.658652305603027\\
73	-16.265127182006836\\
74	-15.225445747375488\\
75	-14.441570281982422\\
76	-13.502510070800781\\
77	-15.962510108947754\\
78	-14.593903541564941\\
79	-15.237396240234375\\
80	-14.850147247314453\\
81	-15.311431884765625\\
82	-16.653350830078125\\
83	-16.1155948638916\\
84	-15.534679412841797\\
85	-15.388496398925781\\
86	-16.122285842895508\\
87	-16.525588989257812\\
88	-13.772363662719727\\
89	-15.156655311584473\\
90	-13.983728408813477\\
91	-14.831045150756836\\
92	-13.514721870422363\\
93	-16.26898765563965\\
94	-15.671761512756348\\
95	-13.933684349060059\\
96	-16.264083862304688\\
97	-14.1621732711792\\
98	-15.292488098144531\\
99	-15.65710735321045\\
100	-16.37059783935547\\
101	-15.006610870361328\\
102	-15.481057167053223\\
103	-15.572470664978027\\
104	-15.948452949523926\\
105	-17.17485237121582\\
106	-15.362051963806152\\
107	-15.635839462280273\\
108	-15.444011688232422\\
109	-15.806888580322266\\
110	-15.910392761230469\\
111	-15.608752250671387\\
112	-15.438133239746094\\
113	-15.273118019104004\\
114	-15.311563491821289\\
115	-14.579501152038574\\
116	-15.286798477172852\\
117	-15.407919883728027\\
118	-16.071447372436523\\
119	-15.709494590759277\\
120	-15.400007247924805\\
121	-13.3629150390625\\
122	-16.10527992248535\\
123	-16.877670288085938\\
124	-16.840057373046875\\
125	-14.958300590515137\\
126	-15.84737777709961\\
127	-15.735976219177246\\
128	-16.35589027404785\\
129	-16.34817886352539\\
130	-15.244711875915527\\
131	-15.550270080566406\\
132	-14.772716522216797\\
133	-15.805510520935059\\
134	-15.505545616149902\\
135	-16.0878963470459\\
136	-14.54876708984375\\
137	-15.845070838928223\\
138	-13.98391342163086\\
139	-14.629054069519043\\
140	-16.709871292114258\\
141	-15.131575584411621\\
142	-14.48392105102539\\
143	-15.491523742675781\\
144	-15.969717979431152\\
145	-15.384284973144531\\
146	-15.157244682312012\\
147	-15.92455768585205\\
148	-15.555231094360352\\
149	-15.326488494873047\\
150	-15.79083251953125\\
151	-14.020951271057129\\
152	-14.80795669555664\\
153	-15.336583137512207\\
154	-15.653914451599121\\
155	-14.551148414611816\\
156	-15.8663969039917\\
157	-17.387781143188477\\
158	-16.76451873779297\\
159	-16.083581924438477\\
160	-15.786467552185059\\
161	-14.322775840759277\\
162	-15.109570503234863\\
163	-15.132270812988281\\
164	-16.16598892211914\\
165	-15.546499252319336\\
166	-15.189489364624023\\
167	-15.692082405090332\\
168	-13.932961463928223\\
169	-13.97856330871582\\
170	-15.264686584472656\\
171	-15.576730728149414\\
172	-15.842275619506836\\
173	-16.269241333007812\\
174	-17.054180145263672\\
175	-15.222312927246094\\
176	-15.57148551940918\\
177	-14.869007110595703\\
178	-16.787925720214844\\
179	-16.505474090576172\\
180	-15.581612586975098\\
181	-15.613595962524414\\
182	-13.096240043640137\\
183	-17.455778121948242\\
184	-15.51548957824707\\
185	-15.896035194396973\\
186	-16.56831169128418\\
187	-14.801456451416016\\
188	-17.392419815063477\\
189	-15.350114822387695\\
190	-14.807945251464844\\
191	-15.454405784606934\\
192	-14.536833763122559\\
193	-14.68761920928955\\
194	-15.623047828674316\\
195	-14.711124420166016\\
196	-15.793560028076172\\
197	-15.018747329711914\\
198	-16.321348190307617\\
199	-15.134486198425293\\
200	-15.337496757507324\\
201	-15.40939712524414\\
202	-16.565353393554688\\
203	-15.54684829711914\\
204	-16.248096466064453\\
205	-14.457324981689453\\
206	-14.970041275024414\\
207	-14.454188346862793\\
208	-15.556659698486328\\
209	-15.887310028076172\\
210	-15.254131317138672\\
211	-15.905055046081543\\
212	-15.17338752746582\\
213	-17.053499221801758\\
214	-14.201851844787598\\
215	-15.533059120178223\\
216	-15.429338455200195\\
217	-16.015823364257812\\
218	-15.871227264404297\\
219	-15.049899101257324\\
220	-16.17119789123535\\
221	-15.557302474975586\\
222	-17.71240234375\\
223	-16.30997085571289\\
224	-16.182239532470703\\
225	-15.758975982666016\\
226	-16.57137107849121\\
227	-15.7284517288208\\
228	-14.740008354187012\\
229	-14.431695938110352\\
230	-16.14474868774414\\
231	-14.622549057006836\\
232	-16.425960540771484\\
233	-14.606647491455078\\
234	-15.994183540344238\\
235	-16.295974731445312\\
236	-14.811108589172363\\
237	-16.031404495239258\\
238	-15.566444396972656\\
239	-16.656293869018555\\
240	-16.80954360961914\\
241	-14.047721862792969\\
242	-15.997150421142578\\
243	-17.042091369628906\\
244	-14.776610374450684\\
245	-16.117307662963867\\
246	-17.73802947998047\\
247	-15.3446044921875\\
248	-14.50361442565918\\
249	-16.652095794677734\\
250	-16.04107093811035\\
251	-16.343549728393555\\
252	-15.900956153869629\\
253	-15.484901428222656\\
254	-16.41907501220703\\
255	-15.733163833618164\\
256	-15.62830924987793\\
257	-14.270946502685547\\
258	-15.433366775512695\\
259	-15.663095474243164\\
260	-13.105605125427246\\
261	-15.721565246582031\\
262	-14.347332954406738\\
263	-15.582023620605469\\
264	-17.077112197875977\\
265	-16.311363220214844\\
266	-17.550254821777344\\
267	-15.870757102966309\\
268	-16.469026565551758\\
269	-16.43752670288086\\
270	-15.556135177612305\\
271	-16.322864532470703\\
272	-17.178590774536133\\
273	-17.148578643798828\\
274	-15.159046173095703\\
275	-17.987611770629883\\
276	-15.686596870422363\\
277	-16.223285675048828\\
278	-14.683649063110352\\
279	-15.031424522399902\\
280	-17.998126983642578\\
281	-14.775103569030762\\
282	-16.5321102142334\\
283	-17.54035186767578\\
284	-16.340852737426758\\
285	-14.784473419189453\\
286	-16.143903732299805\\
287	-17.410602569580078\\
288	-16.078357696533203\\
289	-14.685464859008789\\
290	-16.161901473999023\\
291	-14.867581367492676\\
292	-15.295293807983398\\
293	-17.265548706054688\\
294	-14.103510856628418\\
295	-14.39753532409668\\
296	-16.532310485839844\\
297	-15.164678573608398\\
298	-17.279088973999023\\
299	-15.401711463928223\\
};
\addlegendentry{Nonzero initialization of output-layer weights}

\addplot[
    line width=1pt,
    color=profRed,
]table[row sep=crcr]{%
0	-11.28169059753418\\
1	-13.629807472229004\\
2	-13.252885818481445\\
3	-14.101726531982422\\
4	-12.575460433959961\\
5	-15.366774559020996\\
6	-13.17846965789795\\
7	-12.087149620056152\\
8	-13.233033180236816\\
9	-12.938942909240723\\
10	-13.144323348999023\\
11	-12.076696395874023\\
12	-14.019974708557129\\
13	-13.297263145446777\\
14	-12.125812530517578\\
15	-12.767570495605469\\
16	-13.237926483154297\\
17	-12.472760200500488\\
18	-14.773455619812012\\
19	-14.884446144104004\\
20	-13.158773422241211\\
21	-14.17959213256836\\
22	-14.203335762023926\\
23	-15.504298210144043\\
24	-14.027894020080566\\
25	-18.541370391845703\\
26	-15.818424224853516\\
27	-15.540214538574219\\
28	-14.677263259887695\\
29	-15.735634803771973\\
30	-16.094594955444336\\
31	-16.131258010864258\\
32	-16.4555721282959\\
33	-15.325919151306152\\
34	-14.40942096710205\\
35	-16.648344039916992\\
36	-14.855045318603516\\
37	-14.791826248168945\\
38	-15.584308624267578\\
39	-13.138070106506348\\
40	-14.067163467407227\\
41	-15.635075569152832\\
42	-17.000137329101562\\
43	-15.123806953430176\\
44	-15.424147605895996\\
45	-16.507740020751953\\
46	-16.193517684936523\\
47	-15.881516456604004\\
48	-16.586162567138672\\
49	-16.82452964782715\\
50	-16.047496795654297\\
51	-14.98737907409668\\
52	-15.84195327758789\\
53	-17.0716552734375\\
54	-14.918185234069824\\
55	-15.057441711425781\\
56	-17.569568634033203\\
57	-15.596352577209473\\
58	-17.192838668823242\\
59	-16.924270629882812\\
60	-16.30906867980957\\
61	-16.517141342163086\\
62	-16.826448440551758\\
63	-16.13434600830078\\
64	-17.128555297851562\\
65	-15.991680145263672\\
66	-15.820943832397461\\
67	-17.734683990478516\\
68	-15.984552383422852\\
69	-17.868976593017578\\
70	-14.828207969665527\\
71	-15.24096393585205\\
72	-17.964122772216797\\
73	-18.73337173461914\\
74	-17.39032745361328\\
75	-15.898294448852539\\
76	-15.685629844665527\\
77	-18.068111419677734\\
78	-16.28003692626953\\
79	-17.513301849365234\\
80	-16.8652286529541\\
81	-17.2940731048584\\
82	-19.103416442871094\\
83	-17.96877670288086\\
84	-16.88323402404785\\
85	-18.484264373779297\\
86	-18.159414291381836\\
87	-18.537385940551758\\
88	-15.693676948547363\\
89	-17.61636734008789\\
90	-16.132808685302734\\
91	-17.57901382446289\\
92	-15.898133277893066\\
93	-19.56355094909668\\
94	-18.261272430419922\\
95	-15.841955184936523\\
96	-18.25169563293457\\
97	-17.272432327270508\\
98	-18.623342514038086\\
99	-17.916898727416992\\
100	-19.398265838623047\\
101	-18.173555374145508\\
102	-18.307296752929688\\
103	-17.80292510986328\\
104	-19.083446502685547\\
105	-19.92090606689453\\
106	-17.834781646728516\\
107	-17.988845825195312\\
108	-18.69803237915039\\
109	-19.08523178100586\\
110	-18.969341278076172\\
111	-19.563901901245117\\
112	-18.74555206298828\\
113	-18.32575035095215\\
114	-18.808435440063477\\
115	-16.36037826538086\\
116	-17.579626083374023\\
117	-18.217418670654297\\
118	-17.7843074798584\\
119	-18.398887634277344\\
120	-17.657817840576172\\
121	-15.142492294311523\\
122	-17.70888900756836\\
123	-19.096559524536133\\
124	-18.687131881713867\\
125	-17.437318801879883\\
126	-18.98811149597168\\
127	-18.306486129760742\\
128	-18.066919326782227\\
129	-18.36888885498047\\
130	-17.075111389160156\\
131	-17.8388614654541\\
132	-17.661603927612305\\
133	-19.128681182861328\\
134	-17.649063110351562\\
135	-19.140281677246094\\
136	-16.277191162109375\\
137	-18.163558959960938\\
138	-16.65976333618164\\
139	-17.26721954345703\\
140	-19.715015411376953\\
141	-17.1737060546875\\
142	-16.834932327270508\\
143	-18.5297908782959\\
144	-18.27648162841797\\
145	-18.293563842773438\\
146	-17.647724151611328\\
147	-19.411516189575195\\
148	-18.704681396484375\\
149	-18.830842971801758\\
150	-18.76209831237793\\
151	-15.990727424621582\\
152	-17.74207305908203\\
153	-18.625709533691406\\
154	-18.714475631713867\\
155	-16.745901107788086\\
156	-18.884313583374023\\
157	-20.977712631225586\\
158	-19.78839683532715\\
159	-18.44664192199707\\
160	-18.61786651611328\\
161	-16.39423370361328\\
162	-17.95808219909668\\
163	-17.03105926513672\\
164	-18.358488082885742\\
165	-18.06818389892578\\
166	-18.020814895629883\\
167	-18.615690231323242\\
168	-16.076738357543945\\
169	-15.868721008300781\\
170	-18.579952239990234\\
171	-17.999597549438477\\
172	-18.445629119873047\\
173	-19.34949493408203\\
174	-20.651880264282227\\
175	-18.121471405029297\\
176	-18.425756454467773\\
177	-17.183469772338867\\
178	-20.434144973754883\\
179	-19.238811492919922\\
180	-18.633028030395508\\
181	-18.935434341430664\\
182	-15.280876159667969\\
183	-20.33028793334961\\
184	-18.083715438842773\\
185	-19.691814422607422\\
186	-19.541641235351562\\
187	-17.680383682250977\\
188	-19.93601417541504\\
189	-17.950071334838867\\
190	-18.096694946289062\\
191	-18.256000518798828\\
192	-18.514114379882812\\
193	-17.731216430664062\\
194	-18.661602020263672\\
195	-18.230989456176758\\
196	-18.54991912841797\\
197	-17.528884887695312\\
198	-19.464509963989258\\
199	-17.874011993408203\\
200	-17.767210006713867\\
201	-19.3642635345459\\
202	-20.03325843811035\\
203	-19.07288360595703\\
204	-18.689165115356445\\
205	-16.441747665405273\\
206	-18.271373748779297\\
207	-17.631563186645508\\
208	-19.496789932250977\\
209	-19.441591262817383\\
210	-17.914203643798828\\
211	-18.47783660888672\\
212	-17.96481704711914\\
213	-20.47698402404785\\
214	-16.61884117126465\\
215	-18.1611385345459\\
216	-18.29078483581543\\
217	-19.055862426757812\\
218	-18.746179580688477\\
219	-17.936677932739258\\
220	-19.421682357788086\\
221	-18.42268943786621\\
222	-20.881746292114258\\
223	-19.62796401977539\\
224	-18.4644832611084\\
225	-19.57208824157715\\
226	-19.289701461791992\\
227	-18.229524612426758\\
228	-17.516630172729492\\
229	-17.047954559326172\\
230	-18.929893493652344\\
231	-17.581453323364258\\
232	-19.6486873626709\\
233	-17.233442306518555\\
234	-19.631202697753906\\
235	-20.67949676513672\\
236	-17.685012817382812\\
237	-19.200227737426758\\
238	-18.552345275878906\\
239	-20.23284339904785\\
240	-19.370468139648438\\
241	-16.638561248779297\\
242	-19.36263084411621\\
243	-20.711711883544922\\
244	-16.640329360961914\\
245	-18.864625930786133\\
246	-20.40576934814453\\
247	-18.407106399536133\\
248	-18.288312911987305\\
249	-19.958721160888672\\
250	-18.68463706970215\\
251	-18.691917419433594\\
252	-18.278322219848633\\
253	-18.0230655670166\\
254	-19.139598846435547\\
255	-18.682600021362305\\
256	-18.66814422607422\\
257	-16.643592834472656\\
258	-18.765663146972656\\
259	-19.08013343811035\\
260	-15.978178977966309\\
261	-18.735034942626953\\
262	-17.47173500061035\\
263	-18.87144660949707\\
264	-20.23041343688965\\
265	-19.350418090820312\\
266	-20.597057342529297\\
267	-19.852773666381836\\
268	-19.160568237304688\\
269	-18.8535213470459\\
270	-18.04802131652832\\
271	-18.912691116333008\\
272	-20.833559036254883\\
273	-21.241783142089844\\
274	-17.339799880981445\\
275	-20.815773010253906\\
276	-18.52910614013672\\
277	-18.74208641052246\\
278	-17.662860870361328\\
279	-17.626585006713867\\
280	-20.76950454711914\\
281	-18.47214698791504\\
282	-19.370014190673828\\
283	-20.029075622558594\\
284	-18.950597763061523\\
285	-18.073163986206055\\
286	-19.09868621826172\\
287	-20.080041885375977\\
288	-18.632553100585938\\
289	-17.650917053222656\\
290	-19.895526885986328\\
291	-16.85262107849121\\
292	-17.814918518066406\\
293	-20.318191528320312\\
294	-17.971860885620117\\
295	-16.900882720947266\\
296	-19.864089965820312\\
297	-18.391830444335938\\
298	-20.43206024169922\\
299	-17.629680633544922\\
};
\addlegendentry{Zero initialization of output-layer weights}

\addplot[
    line width=1.5pt,
    densely dotted,
    color=black,
]table[row sep=crcr]{%
0	-11.28169059753418\\
1	-13.683599472045898\\
2	-13.188577651977539\\
3	-14.169621467590332\\
4	-12.582854270935059\\
5	-15.269936561584473\\
6	-13.120899200439453\\
7	-12.148850440979004\\
8	-13.270868301391602\\
9	-13.059840202331543\\
10	-12.909720420837402\\
11	-11.8245267868042\\
12	-13.72885513305664\\
13	-12.650859832763672\\
14	-11.7281494140625\\
15	-11.961021423339844\\
16	-12.374499320983887\\
17	-11.686671257019043\\
18	-13.238247871398926\\
19	-12.864261627197266\\
20	-11.163893699645996\\
21	-12.129725456237793\\
22	-11.681000709533691\\
23	-13.429516792297363\\
24	-10.45287799835205\\
25	-14.712197303771973\\
26	-11.635512351989746\\
27	-12.193621635437012\\
28	-12.407244682312012\\
29	-12.184562683105469\\
30	-11.693891525268555\\
31	-12.585415840148926\\
32	-12.724225044250488\\
33	-12.566303253173828\\
34	-12.218938827514648\\
35	-13.15317153930664\\
36	-12.537917137145996\\
37	-11.80626106262207\\
38	-13.053916931152344\\
39	-10.393903732299805\\
40	-11.800314903259277\\
41	-13.740059852600098\\
42	-13.726794242858887\\
43	-12.611519813537598\\
44	-11.930272102355957\\
45	-13.842411994934082\\
46	-13.01888656616211\\
47	-12.81631088256836\\
48	-13.126733779907227\\
49	-13.009079933166504\\
50	-13.125505447387695\\
51	-11.936694145202637\\
52	-11.969734191894531\\
53	-13.833989143371582\\
54	-12.262393951416016\\
55	-11.913765907287598\\
56	-12.857695579528809\\
57	-12.1857328414917\\
58	-13.016348838806152\\
59	-12.215304374694824\\
60	-12.713418960571289\\
61	-12.037341117858887\\
62	-13.111078262329102\\
63	-12.10064697265625\\
64	-13.054510116577148\\
65	-10.924137115478516\\
66	-12.707606315612793\\
67	-13.645881652832031\\
68	-11.354752540588379\\
69	-13.029162406921387\\
70	-11.75372314453125\\
71	-11.648406982421875\\
72	-13.57433795928955\\
73	-13.313116073608398\\
74	-11.90081787109375\\
75	-11.080146789550781\\
76	-11.657980918884277\\
77	-13.423905372619629\\
78	-12.04853630065918\\
79	-12.911125183105469\\
80	-12.918338775634766\\
81	-13.042569160461426\\
82	-14.255363464355469\\
83	-13.274913787841797\\
84	-12.85964584350586\\
85	-13.270659446716309\\
86	-12.793249130249023\\
87	-13.45898151397705\\
88	-11.969189643859863\\
89	-13.009708404541016\\
90	-11.57662582397461\\
91	-11.564112663269043\\
92	-11.15812873840332\\
93	-12.877741813659668\\
94	-12.666592597961426\\
95	-10.781394958496094\\
96	-12.710662841796875\\
97	-12.327109336853027\\
98	-12.975639343261719\\
99	-13.09505844116211\\
100	-12.825603485107422\\
101	-12.172679901123047\\
102	-13.225090026855469\\
103	-13.368461608886719\\
104	-13.923482894897461\\
105	-13.72989273071289\\
106	-13.030556678771973\\
107	-13.433258056640625\\
108	-12.473557472229004\\
109	-12.814273834228516\\
110	-13.281696319580078\\
111	-13.074990272521973\\
112	-12.52743148803711\\
113	-13.224237442016602\\
114	-12.354327201843262\\
115	-11.700835227966309\\
116	-12.692898750305176\\
117	-12.287408828735352\\
118	-12.028242111206055\\
119	-12.891165733337402\\
120	-12.600354194641113\\
121	-10.722076416015625\\
122	-13.168143272399902\\
123	-13.284721374511719\\
124	-12.80586051940918\\
125	-12.944048881530762\\
126	-14.030884742736816\\
127	-13.207871437072754\\
128	-12.348011016845703\\
129	-12.933292388916016\\
130	-11.968770980834961\\
131	-12.606053352355957\\
132	-11.601797103881836\\
133	-13.443154335021973\\
134	-12.81607437133789\\
135	-14.173192024230957\\
136	-11.421574592590332\\
137	-12.694744110107422\\
138	-11.888554573059082\\
139	-11.383627891540527\\
140	-12.627668380737305\\
141	-12.21296501159668\\
142	-11.676854133605957\\
143	-12.825811386108398\\
144	-12.459569931030273\\
145	-12.096872329711914\\
146	-13.169910430908203\\
147	-12.196084976196289\\
148	-12.52682113647461\\
149	-12.634599685668945\\
150	-12.935782432556152\\
151	-12.460196495056152\\
152	-12.418042182922363\\
153	-11.928404808044434\\
154	-12.304305076599121\\
155	-11.810644149780273\\
156	-13.214275360107422\\
157	-14.249966621398926\\
158	-13.395488739013672\\
159	-12.43443489074707\\
160	-12.531805038452148\\
161	-11.289754867553711\\
162	-13.126542091369629\\
163	-12.776756286621094\\
164	-11.969067573547363\\
165	-12.011212348937988\\
166	-12.503523826599121\\
167	-12.235993385314941\\
168	-11.469085693359375\\
169	-11.748088836669922\\
170	-12.481240272521973\\
171	-13.18829345703125\\
172	-12.860511779785156\\
173	-12.566544532775879\\
174	-13.885297775268555\\
175	-12.664681434631348\\
176	-11.608478546142578\\
177	-12.333898544311523\\
178	-14.383100509643555\\
179	-13.192007064819336\\
180	-13.283376693725586\\
181	-12.291731834411621\\
182	-10.280879020690918\\
183	-14.066900253295898\\
184	-12.19736099243164\\
185	-12.45189094543457\\
186	-13.574450492858887\\
187	-12.070696830749512\\
188	-13.79174518585205\\
189	-12.322408676147461\\
190	-12.34449577331543\\
191	-12.568198204040527\\
192	-11.902966499328613\\
193	-11.944570541381836\\
194	-12.200057029724121\\
195	-11.808389663696289\\
196	-13.169515609741211\\
197	-11.354796409606934\\
198	-12.334992408752441\\
199	-11.817994117736816\\
200	-12.430716514587402\\
201	-13.578544616699219\\
202	-13.288098335266113\\
203	-12.462894439697266\\
204	-13.053821563720703\\
205	-11.512943267822266\\
206	-12.115767478942871\\
207	-11.531194686889648\\
208	-12.558629989624023\\
209	-13.555817604064941\\
210	-13.120327949523926\\
211	-13.254481315612793\\
212	-12.927204132080078\\
213	-13.270747184753418\\
214	-11.264342308044434\\
215	-13.072543144226074\\
216	-12.569931983947754\\
217	-12.778602600097656\\
218	-12.280428886413574\\
219	-12.495423316955566\\
220	-13.132085800170898\\
221	-12.126028060913086\\
222	-14.517221450805664\\
223	-12.816049575805664\\
224	-12.576638221740723\\
225	-12.047296524047852\\
226	-12.672350883483887\\
227	-12.05123519897461\\
228	-11.585274696350098\\
229	-11.89495849609375\\
230	-12.528661727905273\\
231	-12.740270614624023\\
232	-11.529540061950684\\
233	-11.607217788696289\\
234	-12.067741394042969\\
235	-13.3705472946167\\
236	-12.732686996459961\\
237	-12.246196746826172\\
238	-13.516805648803711\\
239	-14.688119888305664\\
240	-12.464781761169434\\
241	-11.276394844055176\\
242	-13.111852645874023\\
243	-13.096508026123047\\
244	-12.597185134887695\\
245	-12.537347793579102\\
246	-14.126446723937988\\
247	-12.00257682800293\\
248	-11.657454490661621\\
249	-13.327062606811523\\
250	-12.444397926330566\\
251	-13.342704772949219\\
252	-13.24150276184082\\
253	-11.159695625305176\\
254	-13.06959342956543\\
255	-12.51374626159668\\
256	-12.18639850616455\\
257	-12.179290771484375\\
258	-12.163237571716309\\
259	-12.753425598144531\\
260	-10.35130500793457\\
261	-12.20419692993164\\
262	-12.144340515136719\\
263	-12.530793190002441\\
264	-13.396689414978027\\
265	-12.822314262390137\\
266	-13.761034965515137\\
267	-13.374345779418945\\
268	-13.81601333618164\\
269	-13.009438514709473\\
270	-12.761037826538086\\
271	-13.524927139282227\\
272	-13.772974014282227\\
273	-13.677367210388184\\
274	-12.447758674621582\\
275	-13.718814849853516\\
276	-11.909775733947754\\
277	-13.965837478637695\\
278	-11.806452751159668\\
279	-11.70071029663086\\
280	-13.759783744812012\\
281	-11.636096954345703\\
282	-13.047187805175781\\
283	-14.235538482666016\\
284	-12.930264472961426\\
285	-11.753798484802246\\
286	-12.7671480178833\\
287	-13.101109504699707\\
288	-12.187082290649414\\
289	-11.857927322387695\\
290	-12.771341323852539\\
291	-11.580976486206055\\
292	-12.361122131347656\\
293	-13.434535026550293\\
294	-11.292545318603516\\
295	-11.658474922180176\\
296	-12.893082618713379\\
297	-11.17505168914795\\
298	-13.361254692077637\\
299	-11.819681167602539\\
};
\addlegendentry{(Negative) optimized WSR via model-based Alg.~\ref{ModelAlgo}}

\end{axis}
\end{tikzpicture}

%% file: figs/fig_WSRvsRuntime.tex
\renewcommand{\sfdefault}{cmbr} 
\newcommand{\ttk}[1]{\footnotesize\textsf{#1}} 

\pgfplotsset{
    every axis/.append style={
        label style={font=\sffamily\footnotesize},      
        legend style={font=\sffamily\footnotesize},     
        tick label style={font=\sffamily\footnotesize}  
    }
}

\begin{tikzpicture}
\begin{axis}[%
    width=2.97in, height=2in, scale only axis,
    xmin=0, xmax=20, ymin=1, ymax=38,
    xmode=log,
    xlabel={Runtime [seconds]}, 
    ylabel={Achieved WSR [bps/Hz]},
    xtick={0.1,0.2,0.3,0.4,0.5,0.6,0.7,0.8,0.9,1,2,3,4,5,6,7,8,9,10,20,30,40,50,60,70,80,90,100,200,300,400},
    ytick={0,5,10,15,20,25,30,35,40},
    xticklabels={{},\ttk{0.2},{},{},\ttk{0.5},{},{},{},{},\ttk{1},\ttk{2},{},{},\ttk{5},{},{},{},{},\ttk{10},{},{},{},{},{},{},{},{},\ttk{100},{},{},{}},
    yticklabels={\ttk{0},\ttk{5},\ttk{10},\ttk{15},\ttk{20},\ttk{25},\ttk{30},\ttk{35},\ttk{40}},
    axis line style={line width=0.7pt, draw=black!40},
    axis background/.style={fill=white},
    xmajorgrids=true,
    ymajorgrids=true,
    grid style={dashed, line width=0.5pt, draw=black!12},
    legend style={
        at={(0.5,1.05)}, anchor=south,
        legend cell align=left, 
        legend columns=3, 
        align=left,
        fill=white,
        draw=black!40,
        line width=0.7pt,
        font=\sffamily\tiny}
]

\addplot[
    only marks,
    mark=*,
    mark size=2pt,
    color=tabblue,
] table[row sep=crcr]{%
2.909279e-01	11.057173\\
1.886899e-01	13.581409\\
1.984808e-01	23.604412\\
1.794889e-01	12.926561\\
1.959748e-01	8.712996\\
1.569259e-01	4.473157\\
1.743660e-01	6.376812\\
1.645262e-01	16.228903\\
1.574059e-01	10.528525\\
2.070768e-01	21.760420\\
2.204030e-01	24.281862\\
1.560299e-01	16.493393\\
2.337599e-01	22.390772\\
1.759818e-01	12.119472\\
1.755872e-01	14.163906\\
1.941092e-01	16.095266\\
};\addlegendentry{Unfolded Alg.~\ref{ModelAlgo} ($L$=2)\,\quad\ }

\addplot[
    only marks,
    mark=triangle*,
    mark size=2.5pt,
    color=taborange,
] table[row sep=crcr]{%
4.500291e-01	11.241055\\
3.613679e-01	21.025366\\
2.639222e-01	26.677519\\
2.751560e-01	16.591639\\
3.104191e-01	9.412502\\
2.231224e-01	4.894503\\
3.205111e-01	11.731529\\
2.768192e-01	14.416709\\
2.281289e-01	10.958114\\
2.650590e-01	29.662352\\
2.659979e-01	32.525482\\
2.737260e-01	18.836733\\
4.011810e-01	27.531694\\
3.216469e-01	12.985290\\
2.621698e-01	18.326509\\
3.022010e-01	19.074768\\
};\addlegendentry{Unfolded Alg.~\ref{ModelAlgo} ($L$=4)\,\quad\ }

\addplot[
    only marks,
    mark=square*,
    mark size=1.6pt,
    color=tabgreen,
] table[row sep=crcr]{%
4.291451e-01	11.234782\\
3.550620e-01	21.747099\\
4.414740e-01	25.196762\\
3.972862e-01	16.683128\\
4.808381e-01	9.652327\\
2.961521e-01	4.929611\\
3.503480e-01	11.784321\\
2.957320e-01	14.899930\\
4.079607e-01	11.335408\\
3.727529e-01	30.430546\\
3.482850e-01	35.457535\\
3.327961e-01	18.215651\\
4.421959e-01	29.872442\\
3.480649e-01	13.005246\\
3.996501e-01	20.502361\\
4.903460e-01	19.015190\\
};\addlegendentry{Unfolded Alg.~\ref{ModelAlgo} ($L$=6)}

\addplot[
    only marks,
    mark=pentagon,
    mark size=2.2pt,
    line width=0.8pt,
    color=tabred,
] table[row sep=crcr]{%
8.244650e+00	11.620625\\
7.075068e+00	18.592163\\
5.430616e+00	22.103212\\
6.666707e+00	17.695824\\
6.360641e+00	10.041546\\
1.668616e+00	4.974675\\
6.788353e+00	14.820833\\
5.379013e+00	15.523195\\
5.385146e+00	11.926088\\
6.679929e+00	32.454399\\
6.653937e+00	24.485802\\
5.332873e+00	20.836136\\
8.271151e+00	31.217968\\
6.741802e+00	13.760427\\
6.781861e+00	18.384701\\
8.016479e+00	18.917418\\
};\addlegendentry{Alg.~\ref{ModelAlgo} ($L$=200)}

\addplot[
    only marks,
    mark=+,
    mark size=2.2pt,
    line width=0.8pt,
    color=black,
] table[row sep=crcr]{%
3.286651e+00	6.407126\\
7.613615e+00	19.410210\\
8.336654e+00	23.037167\\
2.625105e+00	12.620838\\
3.247528e+00	9.493598\\
2.169101e+00	3.141079\\
3.489359e+00	10.104256\\
7.521987e+00	11.043639\\
2.141586e+00	10.434544\\
2.647735e+00	24.361130\\
8.277720e+00	28.502680\\
2.275177e+00	16.653561\\
4.253829e+00	26.412628\\
2.753992e+00	12.019421\\
6.028980e+00	17.676218\\
1.020719e+00	15.873992\\
};\addlegendentry{Benchmark: AO-MO}

\addplot[
    only marks,
    mark=x,
    mark size=2.2pt,
    line width=0.8pt,
    color=black,
] table[row sep=crcr]{%
3.172381e-01	8.127346\\
2.818801e-01	9.741897\\
2.679310e-01	23.002935\\
2.691679e-01	14.466516\\
2.465479e-01	3.729205\\
1.610513e-01	2.851811\\
2.918382e-01	8.822935\\
2.501943e-01	14.372696\\
2.751060e-01	6.752618\\
2.595210e-01	14.936431\\
2.650728e-01	13.335495\\
2.948151e-01	12.048157\\
2.942009e-01	13.333979\\
2.636111e-01	7.944017\\
3.019681e-01	16.558411\\
2.918742e-01	8.763159\\
};\addlegendentry{Benchmark: Mat. Decomp.}

\end{axis}
\end{tikzpicture}

%% file: figs/fig_WSRvsL.tex
\renewcommand{\sfdefault}{cmbr} 
\newcommand{\ttk}[1]{\footnotesize\textsf{#1}} 

\pgfplotsset{
    every axis/.append style={
        label style={font=\sffamily\footnotesize},      
        legend style={font=\sffamily\footnotesize},     
        tick label style={font=\sffamily\footnotesize}  
    }
}

\begin{tikzpicture}
\begin{axis}[
    width=2.97in, height=2.4in, scale only axis,
    xlabel={Number of unfolding layers},
    ylabel={Achieved WSR [bps/Hz]},
    xmin=1.5, xmax=7.5,
    ymin=0, ymax=47.7,
    xmajorgrids=true,
    ymajorgrids=true,
    grid style={dashed, line width=0.5pt, draw=black!12},
    axis line style={line width=0.7pt, draw=black!40},
    axis background/.style={fill=white},
    xtick={2,3,4,5,6,7},
    ytick={0,10,20,30,40},
    xticklabels={$L=2$,$L=3$,$L=4$,$L=5$,$L=6$,$L=7$},
    legend style={at={(0,1)}, anchor=north west, legend cell align=left, align=left, fill=white, draw=black!40,
        line width=0.7pt, font=\sffamily\scriptsize},
    legend image post style={mark=*, mark options={solid, fill=black}}
]

\addplot[
    area legend,
    draw=black,
    fill=black!10
] coordinates {(0,0)};
\addlegendentry{Min--Max Range}

\foreach \x/\min/\max in {
2/1.120569/30.405960,
3/1.197920/39.770519,
4/1.336551/42.601280,
5/1.403853/42.536961,
6/1.414448/43.771473,
7/1.420401/43.549458
}{
    \addplot[
        draw=black,
        fill=black!10,
        line width=0.8pt
    ] coordinates {
        (\x-0.15,\min)
        (\x-0.15,\max)
        (\x+0.15,\max)
        (\x+0.15,\min)
        (\x-0.15,\min)
    };
}
\addlegendimage{
    legend image code/.code={
        \draw[draw=black, fill=black!10] (0cm,-0.1cm) rectangle (0.3cm,0.15cm);
    }
}

\addplot[
    color=black,
    mark=*,
    mark size=2pt,
    line width=1.1pt,
    mark options={solid}
]
table[
    row sep=\\,
    x=x,
    y=mean
]{
x mean min max \\
2 14.945246 1.120569 30.405960 \\
3 17.293691 1.197920 39.770519 \\
4 18.346993 1.336551 42.601280 \\
5 18.967807 1.403853 42.536961 \\
6 19.133330 1.414448 43.771473 \\
7 19.235471 1.420401 43.549458 \\
};\addlegendentry{Mean}

\end{axis}
\end{tikzpicture}

%% file: figs/fig_WSRvsNMSE.tex
\renewcommand{\sfdefault}{cmbr} 
\newcommand{\ttk}[1]{\footnotesize\textsf{#1}} 

\pgfplotsset{
    every axis/.append style={
        label style={font=\sffamily\footnotesize},      
        legend style={font=\sffamily\footnotesize},     
        tick label style={font=\sffamily\footnotesize}  
    }
}

\begin{tikzpicture}
\begin{axis}[%
    width=2.97in, height=2.4in, scale only axis,
    xmin=0.000085, xmax=0.58, ymin=4.6, ymax=25,
    xmode=log,
    xlabel={CSI NMSE}, 
    ylabel={Achieved WSR [bps/Hz]},
    ytick={0,5,10,15,20,25},
    yticklabels={\ttk{0},\ttk{5},\ttk{10},\ttk{15},\ttk{20},\ttk{25},\ttk{30},\ttk{35},\ttk{40}},
    axis line style={line width=0.7pt, draw=black!40},
    axis background/.style={fill=white},
    xmajorgrids=true,
    ymajorgrids=true,
    grid style={dashed, line width=0.5pt, draw=black!12},
    legend style={
        at={(1,1)}, anchor=north east,
        legend cell align=left,  
        align=left,
        fill=white,
        draw=black!40,
        line width=0.7pt,
        font=\sffamily\scriptsize}
]

\addplot[black,
    dashed,
    mark=pentagon,
    mark size=2.7pt,
    line width=1.2pt,
    mark options={solid, fill=black}
    ]table[row sep=crcr]{%
1.000000e-04	16.838485\\
5.492803e-04	16.825906\\
3.017088e-03	16.806702\\
1.657227e-02	16.654113\\
9.102821e-02	15.522062\\
5.000000e-01	12.179310\\
};
\addlegendentry{Alg.~\ref{ModelAlgo} ($L$=200)}

\addplot[
    black,
    dashed,
    mark=square,
    mark size=2pt,
    line width=1.2pt,
    mark options={solid, fill=black}
] table[row sep=crcr]{%
1.000000e-04	14.696149\\
5.492803e-04	14.684990\\
3.017088e-03	14.691590\\
1.657227e-02	14.689816\\
9.102821e-02	14.070085\\
5.000000e-01	11.352784\\
};
\addlegendentry{Alg.~\ref{ModelAlgo} ($L$=50)}

\addplot[color=taborange, mark=triangle*, mark size=2.2pt, line width=1.2pt, mark options={solid, fill=taborange}] table[row sep=crcr]{%
1.000000e-04	16.298246\\
5.492803e-04	16.271884\\
3.017088e-03	16.025679\\
1.657227e-02	14.912578\\
9.102821e-02	12.406257\\
5.000000e-01	9.123111\\
};
\addlegendentry{Unfolded Alg.~\ref{ModelAlgo} ($L$=4)}

\addplot[color=tabblue, mark=*, mark size=2pt, line width=1.2pt, mark options={solid, fill=tabblue}] table[row sep=crcr]{%
1.000000e-04	13.190988\\
5.492803e-04	13.164840\\
3.017088e-03	13.062592\\
1.657227e-02	12.512697\\
9.102821e-02	10.986320\\
5.000000e-01	8.576329\\
};
\addlegendentry{Unfolded Alg.~\ref{ModelAlgo} ($L$=2)}

\addplot[name path=Unfolded_Iter2_upper, draw=none] table[row sep=crcr]{%
1.000000e-04	19.467990\\
5.492803e-04	19.416836\\
3.017088e-03	19.206744\\
1.657227e-02	18.197700\\
9.102821e-02	15.728093\\
5.000000e-01	12.286233\\
};
\addplot[name path=Unfolded_Iter2_lower, draw=none] table[row sep=crcr]{%
1.000000e-04	6.913985\\
5.492803e-04	6.912845\\
3.017088e-03	6.918440\\
1.657227e-02	6.827694\\
9.102821e-02	6.244546\\
5.000000e-01	4.866425\\
};
\addplot[tabblue!30, opacity=0.5] fill between[of=Unfolded_Iter2_upper and Unfolded_Iter2_lower];

\addplot[name path=Unfolded_Iter4_upper, draw=none] table[row sep=crcr]{%
1.000000e-04	24.473853\\
5.492803e-04	24.417251\\
3.017088e-03	23.892024\\
1.657227e-02	21.838721\\
9.102821e-02	17.673588\\
5.000000e-01	12.871678\\
};
\addplot[name path=Unfolded_Iter4_lower, draw=none] table[row sep=crcr]{%
1.000000e-04	8.122638\\
5.492803e-04	8.126517\\
3.017088e-03	8.159334\\
1.657227e-02	7.986435\\
9.102821e-02	7.138926\\
5.000000e-01	5.374545\\
};

\addplot[taborange!30, opacity=0.5] fill between[of=Unfolded_Iter4_upper and Unfolded_Iter4_lower];

\end{axis}
\end{tikzpicture}

%% file: figs/fig_Generalization.tex
\renewcommand{\sfdefault}{cmbr} 
\newcommand{\ttk}[1]{\footnotesize\textsf{#1}} 

\pgfplotsset{
    every axis/.append style={
        label style={font=\sffamily\footnotesize},      
        legend style={font=\sffamily\footnotesize},     
        tick label style={font=\sffamily\footnotesize}  
    }
}

\begin{tikzpicture}
\begin{axis}[
	width=2.97in, height=2in, scale only axis,
    symbolic x coords={ Miami, Austin, Columbus, Dallas },
    xtick=data,
    ylabel={Average WSR [bit/s/Hz]},
    ymin=0, ymax=22,
    grid style={dashed, line width=0.5pt, draw=black!12},
    axis line style={line width=0.7pt, draw=black!40},
    axis background/.style={fill=white},
    xmajorgrids=true,
    ymajorgrids=true,
    ybar,
    bar width=4pt,
    enlarge x limits=0.2,   
    x tick label style={rotate=0, anchor=center, yshift=-4pt},
    legend style={at={(0.5,1.05)}, anchor=south,  legend cell align=left, align=left, legend columns=3, font=\sffamily\tiny, fill=white, draw=black!40,
        line width=0.7pt},
    ymajorgrids,
]



\addplot[bar shift=-15pt, fill=tabblue!75, draw=tabblue!90!black]
coordinates {(Miami, 14.722206)
(Austin, 15.955746)
(Columbus, 11.943252)
(Dallas, 13.611118)};
\addlegendentry{Unfolded Alg.~\ref{ModelAlgo} ($L$=2)}

\addplot[bar shift=-9pt, fill=taborange!75, draw=taborange!90!black]
coordinates {(Miami, 18.298977)
(Austin, 20.126167)
(Columbus, 14.620319)
(Dallas, 16.246046)};
\addlegendentry{Unfolded Alg.~\ref{ModelAlgo} ($L$=4)}

\addplot[bar shift=-3pt, fill=tabgreen!75, draw=tabgreen!90!black]
coordinates {(Miami, 18.970823)
(Austin, 21.327838)
(Columbus, 15.167595)
(Dallas, 16.709577)};
\addlegendentry{Unfolded Alg.~\ref{ModelAlgo} ($L$=6)}


\addplot[bar shift=3pt, fill=black!65, draw=black]
coordinates {(Miami, 17.335399)
(Austin, 21.363023)
(Columbus, 13.007874)
(Dallas, 16.528502)};
\addlegendentry{Alg.~\ref{ModelAlgo} ($L$=200)}


\addplot[bar shift=9pt, fill=white, draw=black, pattern=north east lines]
coordinates {(Miami, 13.550730)
(Austin, 15.077205)
(Columbus, 10.308733)
(Dallas, 11.760599)};
\addlegendentry{Benchmark: AO-MO}

\addplot[bar shift=15pt, fill=white, draw=black, pattern=north west lines]
coordinates {(Miami, 10.978122)
(Austin, 11.367092)
(Columbus, 9.253147)
(Dallas, 8.572376)};
\addlegendentry{Benchmark: Mat. Decomp.}

\end{axis}
\end{tikzpicture}